\documentclass[12pt]{article}
\usepackage[cmex10]{amsmath}
\usepackage{amssymb}
\usepackage{amsthm}
\usepackage{graphicx}
\usepackage{fancyhdr}
\usepackage{tikz}

\usepackage{enumerate}
\usepackage{hyperref}

\def\eps{\varepsilon}

\newtheorem{thm}{Theorem}[section]
\newtheorem{corollary}[thm]{Corollary}
\newtheorem{lemma}[thm]{Lemma}

\newtheorem{assumption}[thm]{Assumption}

\newtheorem{pr}[thm]{Proposition}

\newtheorem{definition}[thm]{Definition}

\newenvironment{defi}{\begin{definition}\rm}{\end{definition}}
\newtheorem{example}[thm]{Example}

\newtheorem{remark}[thm]{Remark}

\newtheorem{tab}{Table}

\newcommand{\EE}{{\bf E}}
\newcommand{\X}{{\underline{\mathbf{X}}}}
\newcommand{\Z}{{\underline{\mathbf{Z}}}}
\title{On the Capacity of Channels with Deletions and States}
\author{Yonglong Li and Vincent Y. F. Tan \thanks{Y.~Li is  with the Department of Electrical and Computer Engineering,
National University of Singapore (NUS), Singapore 117583 (e-mail: elelong@nus.edu.sg).
V.~Y.~F.~Tan is  with the with the Department of Electrical and Computer Engineering and the Department of Mathematics, NUS, Singapore 119076 (e-mail: vtan@nus.edu.sg).}}

\begin{document}
\maketitle
\begin{abstract}
We consider the class of channels formed from the concatenation of a deletion channel and a finite-state channel. For this class of channels, we show that the operationally-defined capacity is equal to the stationary capacity, which can be approached by a sequence of Markov processes with increasing Markovian orders. As a by-product, we show that the polar coding scheme constructed by Tal,~Pfister,~Fazeli, and Vardy [arxiv: 1904.13385 (2019)] achieves the capacity of the deletion channel.
\end{abstract}
\section{Introduction}
In some communication systems, synchronizing errors form one of the major sources of noise. Paradigmatic examples of this class of channels include insertion and deletion channels. In~\cite{dubrushin67}, Dobrushin considered discrete memoryless channels with synchronizing errors and established the channel coding theorem. In~\cite{ahlswede}, Ahlswede and Wolfowitz proved a channel coding theorem together with an accompanying strong converse for discrete memoryless channels with bounded synchronizing errors. 

In the last two decades, computable bounds on the capacity of deletion and insertion channel have been proposed. In~\cite{surveydeletion, duman}, the authors derived bounds on the capacity of deletion channel. In~\cite{kavcic2010, kavcic2015}, the authors proposed trellis-based approach to numerically compute achievable rates for insertion and deletion channels and also the concatenation of an insertion or deletion channel and a Gaussian inter-symbol interference channel. In~\cite{sudan, montanari}, two group of researchers independently derived the asymptotic behavior of the capacity of the binary deletion channel in the high signal-to-noise regime.

 \begin{figure}
  \begin{center}
\begin{tikzpicture}
\draw (0,0) rectangle (3,1.5);
\draw (5,0) rectangle (8,1.5);
\draw[->]  (3,0.75)--(4,0.75);
\draw[->]  (-1.5,0.75)--(0,0.75);
\draw[-]  (8,0.75)--(9.4,0.75);
\draw  (4,0.75)--(5,0.75);
\node[below] at (1.5,1.1) {\small Deletion Channel};
\node[below] at (6.5,1.1) {FSC};
\node[below] at (-1.0,1.3) {$x_1^n$};
\node[below] at (4.0,1.3) {$Y(x_1^n)$};
\node[below] at (8.7,1.3) {$Z(x_1^n)$};
\end{tikzpicture}
\end{center}
 \caption{Concatenation of a deletion channel and an FSC}
\end{figure}
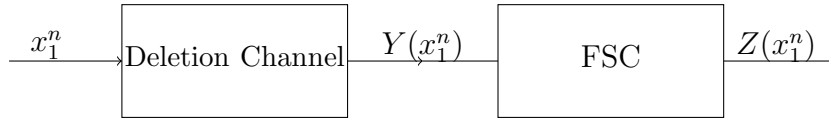
In this work, we consider the class of channels formed from the concatenation of a deletion channel and a finite-state channel (FSC) as shown in Figure $1$. We study the properties of its {\em operational capacity}, which is the largest rate below which reliable communication is possible. One practical reason to consider such a communication system is that   there may be some synchronization errors~\cite{hu2007}  introduced during the writing  process of certain bit-patterned magnetic recording (BPMR) systems. In addition communication channels---for example, the Gaussian inter-symbol interference channels---used to model BPMR systems often possess memory. This motivates us to study the concatenation of a deletion channel and an FSC.

This channel model belongs to the class of channels with deletion and states, whose capacity, in general,  does not admit a closed-form formula. Computing the capacity of channels with deletions or states is a long-standing open problem in information theory. One well-known approach to numerically computing the channel capacity is the so-called {\em Markov approximation} scheme, which has been extensively exploited in past decades for computing the capacity of families of FSCs (see~\cite{arnoldsimulationbounds, vontobel,randomapproachhan} and references therein). The Markov approximation scheme is an approach to compute the lower bounds of the channel capacity by numerically optimizing over the Markovian input processes of order $m$ to obtain the so-called $m^{\mathrm{th}}$-order {\em Markov capacity}. Unlike general input processes, the Markov structure of the input process usually ensures that the computation or approximation of the capacity becomes rather efficient. The effectiveness of this approach has been justified for different channel models in~\cite{chensiegel, lihanflash}, where the authors showed that as the Markov order $m$ tends to infinity, the sequence of the Markov capacities increases to the operationally-defined capacity of the corresponding channel. It is certainly plausible that the Markov approximation scheme can be applied to other channels with memory; as a matter of fact, the main result of the present paper is to confirm this for our channel model.

In the last decade, much progress has been made in computing the Markov capacity of FSCs; in particular, a generalized Blahut-Arimoto algorithm and a randomized algorithm have been respectively proposed in~\cite{vontobel} and~\cite{randomapproachhan}. Although the convergence of the algorithms to the Markov capacity for our channel model is yet unclear, we will justify the effectiveness of the Markov approximation scheme for our channel model. In particular, we show that the sequence of Markov capacities approaches the operationally-defined capacity (which is also the stationary capacity as well as the Shannon capacity) as the Markovian order increases.

Yet another motivation for our work stem from the recent developments in polarization theory for channels with memory. In particular, recently, Shuval and~Tal~\cite{idoFSC}, and Tal,~Pfister,~Fazeli, and Vardy~\cite{idodeletion}, respectively, showed that the strong polarization phenomenon holds for both indecomposable FSCs and deletion channels with regular hidden-Markov inputs.\footnote{This class of inputs will be defined formally in Section~\ref{main}.} Using this fact, they constructed polar coding-based schemes whose rates approach the mutual information between the input and output as the number of polarization levels goes to infinity. However, the fact that regular hidden Markov inputs can approach the capacity was not shown in~\cite{idoFSC,idodeletion}. Hence the authors could not unequivocally conclude that their polar coding scheme  approaches the capacity of the channel. In this paper, we answer---in the  affirmative---the question on whether regular hidden Markov inputs can approach the capacity of a channel formed from the concatenation of a deletion channel and an FSC.

%One technical difficulty to prove the main result lies in the proof of the achievability part.  Due to the deletion, one can not apply the generalized Shannon-Bremain-McMillian theorem~\cite{barron} directly, we will define a new indecomposable FSC.  

The rest of this paper is structured as follows. In Section~\ref{model}, we introduce our channel model and provide a precise description of the problem. In Section~\ref{main}, we state the main results. Before we present the detailed proof of our main results in Section~\ref{proof}, we derive several properties of our channel model in~Section~\ref{property}.
%Let $W_1$ be the channel with timing synchronization errors defined by
%$$
%\Pr(y_1^k|x_1^n)=\sum_{\bar{y}_1^n}\prod_{i=1}^{n}\Pr(\bar{y}|x_i)
%$$
%where the summation is taken over all $\bar{y}_1^n$ such that the concatenation of $\bar{y}_1^n$ is exactly $y_1^k$. 

\section{Problem Setting}\label{model}
Let $\mathcal{X\ ,\mathcal{S},\ \mbox{and}\ Z}$ be finite sets with cardinalities $|\mathcal{X}|$, $|\mathcal{S}|$, and $|\mathcal{Z}|$ respectively. For any positive integer $n$, let $\mathcal{X}^{*n}$ be the set of all vectors over $\mathcal{X}$ with lengths no larger than $n$. For any $x_1^n\in \mathcal{X}^{n}$ and $x^{*}\in\mathcal{X}^{*n}$, let $K(x_1^n,x^{*})$ be the number of ways of producing $x^*$ by possibly deleting some symbols in $x_1^n$. Let $\ell(x^*)$ be the length of $x^{*}$. Let ${W}_1$ be a {\em deletion channel} with deletion probability $d$. As such, the probability of obtaining output $x^*$ when passing the input $x_1^n$ through $W_1^{n}$ is given by  
$$
W_1^n(x^*|x_1^n)=(1-d)^{\ell(x^*)}d^{n-\ell(x^*)}K(x_1^n,x^*).
$$ 
Let $W_2$ be an FSC with input alphabet $\mathcal{X}$, output alphabet $\mathcal{Z}$, and state alphabet $\mathcal{S}$.  The probability of a given output sequence $z_1^n$  and a state sequence $s_1^n$ given an input sequence $x_1^n$ and an initial state $s_0$ is defined as
$$
W_2^{n}(z_1^n,s_1^n|x_1^n,s_0)=\prod_{i=1}^{n}p(z_i,s_i|s_{i-1},x_i),
$$
where $p(z,s|s^\prime,x)$ is a conditional probability mass function, that is, given any $(s^\prime,x)\in\mathcal{S}\times\mathcal{X}$, $p(z,s|s^\prime,x)\ge 0$ and $\sum_{z,s}p(z,s|s^\prime,x)=1$.
Thus the probability of obtaining an output sequence $z_1^n$ by passing a sequence $x_1^n$ through the channel $W_2^n$ with an initial state $s_0$ is given by
$$
W_2^{n}(z_1^n|x_1^n,s_0)=\sum_{s_1^n}W_2^{n}(z_1^n,s_1^n|x_1^n,s_0)=\sum_{s_1^n}\prod_{i=1}^{n}p(z_i,s_i|s_{i-1},x_i).$$
In this paper we only consider {\em indecomposable} FSCs as defined in Gallager's book~\cite[pp.~106]{gallagerbook}. Let $p(s_n|x_1^n,s_0)=\sum_{s_1^{n-1},z_1^n}W_2^{n}(z_1^n,s_1^n|x_1^n,s_0)$ be the conditional marginal probability mass function of $S_n$ given an input sequence $x_1^n$ and an initial state $s_0$. An FSC is said to be {\em indecomposable} if for any $\eps>0$, there exists an $N$ such that for all $n\ge N$
\begin{align}
|p(s_n|x_1^n,s_0)-p(s_n|x_1^n,s_0^\prime)|\le \eps
\end{align}
for any $s_0$, $s_0^{\prime}$, $s_n$, and $x_1^n$.

Let $x_1^n$ and $Y(x_1^n)$ be respectively, the input and output of the deletion channel $W_1^n$. Let $W^n$ be the concatenation of $W_1^n$ and $W_2^{\ell(Y(x_1^n))}$, that is, the output $Y(x_1^n)$ of $W_1^n$ is fed to $W_2^{\ell(Y(x_1^n))}$ as its channel input. Formally, $W^n$ is defined by the conditional probability of observing output $z^{*}\in\mathcal{Z}^{*n}$ when passing $x_1^n$ through $W^n$ with a fixed initial state $s_0$; this is
$$
W^n(z^{*}|x_1^n,s_0)=\sum_{x^{*}\in\mathcal{X}^{*n}:\ell(x^{*})=
\ell(z^*)}W_1^n(x^{*n}|x_1^n)W_2^{\ell(z^{*})}(z^{*}|x^{*},s_0).
$$

Throughout the paper we fix the initial state of $W_2$ to be an arbitrary state $s_0\in\mathcal{S}$ and denote the output of $W^n$ corresponding to the input $x_1^n$ by $Z(x_1^n)$. The state of $W^n$ refers to that of $W_2^n$.

\begin{defi}
An {\em $(n,2^{nR},\eps_n)$-code} with rate $R$ for $W$ is defined by
\begin{itemize}
\item Encoder $f$: a map from $\{1,\cdots,2^{nR}\}$ to $\mathcal{X}^n$;
\item Decoder $g$: a map from $\mathcal{Z}_n^{*}$ to $\{1,\cdots, 2^{nR}\}$;
\item Average error probability $\frac{1}{2^{nR}}\sum_{i=1}^{2^{nR}}\Pr(g(Z(f(i)))\not=i)\le \eps_n$, where $Z(f(i))$ is the output of channel $W^n$ obtained by passing the codeword $f(i)$ through the channel $W^n$.
\end{itemize}
\end{defi}
\begin{defi}
The rate $R$ is said to be {\em achievable} if there exists a sequence of $(n,2^{nR},\eps_n)$-codes for $W^n$ with $\eps_n\to 0$ as $n\to \infty$.
\end{defi}
Let $W=\{W^n\}_{n=1}^{\infty}$ and let $C=\sup\{R: R\ \mbox{is achievable}\}$ be the {\em operational capacity} of channel $W$. In this work, we will show that $C$ can be characterized by several other information capacities. The first such quantity is the {\em Shannon capacity} $$
C_{\mathrm{Shannon}}=\lim_{n\to \infty}\frac{1}{n}\sup_{p_{X_1^n}(\cdot)}I(X_1^n;Z(X_1^n)|s_0),
$$
where $I(X_1^n;Z(X_1^n)|s_0)$ is the mutual information between the input $X_1^n$ and the output $Z(X_1^n)$ when the initial state of $W_2^n$ is fixed to $s_0$.
The second is the {\em stationary capacity} $C_{\mathrm{S}}$ defined as
$$
C_{\mathrm{S}}=\sup_{X}\lim_{n\to\infty}\frac{1}{n}I(X_1^n;Z(X_1^n)|s_0),
$$
where the supremum is taken over all stationary and ergodic input processes $X$.
The final quantity of interest is the {\em $m^{\mathrm{th}}$-order Markov capcity} $C_{\mathrm{Markov}}^{(m)}$ defined as 
$$
C_{\mathrm{Markov}}^{(m)}=\sup_{X}\lim_{n\to\infty}\frac{1}{n}I(X_1^n;Z(X_1^n)|s_0),
$$
where the supremum is taken over all stationary $m^{\mathrm{th}}$-order Markov processes $X$.
\begin{remark}
Although the initial state of $W_2$ is fixed to be $s_0$, all the capacity functions do not depend on $s_0$. This fact will be justified in Corollary~\ref{initialindependent} and Theorem~\ref{shannonindependent}.
\end{remark}

Throughout the paper notations like $p_{X}(x)$ will be used to denote the probability of $X=x$ and similar notations will also be used for conditional probabilities. For a pair of random vectors $(X,Y)$ taking values in $\mathcal{X}\times\mathcal{Y}$, define the information density to be $\iota_{X,Y}(x,y)=\log\frac{p_{X,Y}(x,y)}{p_{X}(x)p_{Y}(y)}$ for $(x,y)\in \mathcal{X}\times\mathcal{Y}$. Sometimes, for notational convenience, we abbreviate $\iota_{X,Y}(X,Y)$ as $\iota_{X,Y}$. The notation $N_{m}(t)\triangleq N(X_{m}^{m+t-1})$ will be used to denote the number of deletions that occurred during the transmission of $X_{m}^{m+t-1}$. When $m=1$, we abbreviate $N_{m}(t)$ as $N(t)$. For a vector $(x_1^{*},\cdots,x_{k}^{*})$, we write $x_1^{*}\cdots x_{k}^{*}$ to denote the concatenation of its symbols.

%Sometimes to emphasize the dependence of the conditional probability $p(y|x)$ on the channel $W$, we use $W(y|x)$ to denote the $p(y|x)$. 

\section{Main Results}\label{main}
Our main contribution is the following theorem:
\begin{thm}\label{capacity} 
For the channel $W$, the following holds:
\begin{align}\label{informationcapacity}
C=C_{\mathrm{Shannon}}=C_{\mathrm{S}}=\lim_{m\to\infty}C_{\mathrm{Markov}}^{(m)}.
\end{align}
\end{thm}
A direct consequence of Theorem~\ref{capacity} is that a sequence of Markov processes with increasing orders asymptotically achieves the capacity of the concatenated channel $W$. In particular, Markov processes with increasing orders achieve the capacity of the deletion channel $W_1$. 

When $|\mathcal{X}|=2$, $W_1$ is known as a {\em binary deletion channel}. A process $X$ is said to be {\em regular hidden-Markov} if both $(X,S)$ and $S$ are stationary irreducible and aperiodic Markov processes. In~\cite{idodeletion}, for a binary deletion channel $W_1$ driven by a regular hidden-Markov input process $X$, the authors showed that the strong polarization phenomenon holds for $(X,Y(X))$. Using this fact, they constructed a sequence of coding schemes $\{\mathcal{C}_k\}$ for the input $X$ with codeword lengths $\{2^k\}$ (i.e.,\ $k$ denotes the levels of the polar transform) and rates $\{R_k\}$. They proved the following theorem:  
\begin{thm}\cite{idodeletion}\label{idopolar}
Fix a regular hidden-Markov input process $X$. For any fixed $\gamma\in(0,1/3)$ and arbitrary $\eps>0$, there is an $N$ such that for polarization level $k\ge N$, the rate of the code $R_k\ge  I(X;Y(X))-\eps$  and the probability of decoding error of $\{\mathcal{C}_k\}$ is upper bounded by $2^{-2^{k\gamma }}$.
\end{thm}

As can be seen from the definition of a regular hidden-Markov process, an irreducible and aperiodic Markov process is regular hidden-Markov. In the following corollary, we show that the sequence of Markov processes that asymptotically achieves $C$ in Theorem~\ref{capacity} can be chosen to be irreducible and ergodic. 

\begin{corollary}\label{polar}
Let $C$ be the capacity of the deletion channel $W_1$. Then for any $\eps>0$, there is an irreducible and aperiodic Markov process $X$ and an integer $M(X, \eps)$ such that for $k\ge M(X,\eps)$, the rate $R_k$ of the coding scheme in~\cite{idodeletion} associated with the input process $X$ satisfies $R_k\ge C-\eps$. 
\end{corollary}
\begin{remark}
One implication of Corollary~\ref{polar} is that the coding scheme constructed in~\cite{idodeletion} is capacity-achieving.
\end{remark}
\begin{remark}
One can apply similar ideas from the proof of Corollary~\ref{polar} to the indecomposable FSC $W_2$ to show that there is a sequence of stationary, irreducible, and aperiodic Markov chains that approaches the capacity of $W_2$ as the Markovian orders increases to infinity.
\end{remark}

\section{Preliminaries}\label{property}
In this section we first derive several important properties concerning our channel model. The first important property is that $W$ is ``indecomposable''.
\begin{lemma}\label{indecomposable}
 Let $S_{k-N(k)}$ be the channel state of $W$ after the transmission of $x_1^k$. Then for any $\eps>0$, there exists an integer $K(\eps)$ such that if $k\ge K(\eps)$, we have 
\begin{align}\label{indecomposability}
\sup_{s\in\mathcal{S}}|\Pr(S_{k-N(k)}=s|X_1^k=x_1^k,S_0=s_0)-\Pr(S_{k-N(k)}=s|X_1^k=x_1^k,S_0=s_0^\prime)|\le \eps,
\end{align}
for any initial states $s_0$, $s_0^{\prime}$, and any input $x_1^k$.
\end{lemma}
\begin{proof}
The proof can be found in Appendix~\ref{indecomposableproof}.
\end{proof}
Throughout the rest of the paper for any given $\eps>0$, $K(\eps)$ is always chosen such that~(\ref{indecomposability}) holds. The following lemma is from~\cite[Lemma 1, pp.~112]{gallagerbook} and is used several times in the proof. For easy reference, we include it as follows.
\begin{lemma}\cite{gallagerbook}\label{conditionallemma}
Let $(X,Y, Z, S)$ be a random vector over $\mathcal{X}\times\mathcal{Y}\times\mathcal{Z}\times\mathcal{S}$, where $\mathcal{X},\ \mathcal{Y},\ \mathcal{Z},\ $ and $\mathcal{S}$ are all finite sets. Then
$$
|I(X;Y|Z,S)-I(X;Y|Z)|\le \log|\mathcal{S}|.
$$
\end{lemma}
The following lemma says that mutual information does not increase too much when additional information about the numbers of deletions over different long blocks is available at the receiver.

\begin{lemma}\label{sideinformation}
Let $m$ and $n$ be positive integers and $\{t_i:0\le i\le m\}$ be a set of integers such that $t_0=0<1\le t_1<t_2<\cdots <t_m=n$. Then 
{\small\begin{align}\label{sideinformationequation}
0\le  I(X_1^{n};Z(X_1^n),N(t_1),\cdots,N_{t_{m-1}+1}(t_{m}-t_{m-1})|s_0)-I(X_1^{n};Z(X_1^n)|s_0)\le \sum_{i=1}^{m}\log(t_{i}-t_{i-1}+1).
\end{align}}
\end{lemma}
\begin{proof}
The proof can be found in Appendix~\ref{sideinformationproof}.
\end{proof}
The following proposition says that for $W$, the difference between the normalized mutual informations given two different initial states is small . 
%\begin{pr}\label{blockdifferenceindependent}
%Let $X_1^n$ and $X_{n+1}^{2n}$ be independent and have the same distribution. Then for $\eps>0$ and $n\ge k\ge  K(\eps)$,
%{\small\begin{align}
%&\frac{|I(X_1^n;Z(X_1^n)|s_0)-I(X_{n+1}^{2n};Z(X_{n+1}^{2n})|S_{n-N(n)})|}{n}\notag\\
%&\hspace{2cm}\le \frac{2(\log|\mathcal{S}|+\log(k+1)+k\log|\mathcal{X}|)+(n-k)\eps\log|\mathcal{X}|}{n}.
%\end{align}}

%\end{pr}
\begin{pr}\label{blockdifferenceindependent}
Let $\eps>0$ be given. Then for any two different initial states $s_0$ and $s_0^\prime$ and for all $n\ge k\ge  K(\eps)$,
{\small\begin{align}
&\frac{|I(X_1^n;Z(X_1^n)|s_0)-I(X_{1}^{n};Z(X_{1}^{n})|s_0^\prime)|}{n}\notag\\
&\hspace{2cm}\le \frac{2(\log|\mathcal{S}|+\log(k+1)+k\log|\mathcal{X}|)+(n-k)\eps\log|\mathcal{X}|}{n}.
\end{align}}

\end{pr}
\begin{proof}
The proof can be found in Appendix~\ref{blockdifferenceindependentproof}.
\end{proof}
\begin{corollary}\label{existenceindependent}
Let $p_{n}(\cdot)$ be a probability mass function on $\mathcal{X}^n$ and $X=\{X_k\}$ be a block independent process with $p_{{X}_{1}^{kn}}(x_1^{kn})=\prod_{i=1}^{k}p_{n}(x_{(i-1)n+1}^{in})$ for any $k$ and $x_{1}^{kn}$. For any $1\le j\le n-1$ and $i\ge 1$, let $X^{(j)}_i\triangleq X_{i+j}$ and $X^{(j)}=\{X^{(j)}_i\}_{i=1}^{\infty}$. Let $Z(X)$ and $Z(X^{(j)})$ be respectively the outputs obtained by passing $X$ and $X^{(j)}$ through the channel $W$. Then 
\begin{itemize}
\item[(i)] For any initial state $s_0$, any $\eps>0$, and all $n\ge k\ge  K(\eps)$,
{\small\begin{align}
&\frac{|I(X_1^n;Z(X_1^n)|s_0)-I(X_{n+1}^{2n};Z(X_{n+1}^{2n})|S_{n-N(n)})|}{n}\notag\\
&\hspace{2cm}\le \frac{2(\log|\mathcal{S}|+\log(k+1)+k\log|\mathcal{X}|)+(n-k)\eps\log|\mathcal{X}|}{n},
\end{align}}
where $S_{n-N(n)}$ is defined as in Lemma~\ref{indecomposable}.
\item[(ii)] For any initial state $s_0$, 
\begin{align}
I(X;Z(X)|s_0)=\lim_{m\to\infty}\frac{1}{m}I(X_1^m;Z(X_1^m)|s_0) \quad\mbox{exists}.
\end{align}  
\item[(iii)] For any initial state $s_0$ and any $1\le j\le n-1$, one has $I(X^{(j)};Z(X^{(j)})|s_0)=I(X;Z(X)|s_0)$.
\item[(iv)] For any initial state $s_0$ and any $\eps>0$, there exists an integer $N(|\mathcal{X}|,|\mathcal{Z}|,|\mathcal{S}|,\eps)$ such that if $n\ge N(|\mathcal{X}|,|\mathcal{X}|,|\mathcal{X}|,\eps)$, then
$$
\left|\lim_{k\to\infty}\frac{I(X_1^{(k)n};Z(X_1^n),\cdots,Z(X_{(k-1)n+1}^{kn})|s_0)}{kn}-\frac{I(X_1^n;Z(X_1^n)|s_0)}{n}\right|\le \eps.
$$
\end{itemize}
\end{corollary}
\begin{proof}
The proof can be found in Appendix~\ref{existenceindependentproof}.
\end{proof}
The following proposition says that for $W$ with a stationary input $X$, the difference between the normalized mutual informations over two long blocks is small.
\begin{pr}\label{blockdifferencestationary}
Let $X=\{X_i\}$ be a stationary input process. Then for any $\eps>0$, any positive integer $k$, and $n\ge t\ge K(\eps)$,  we have that
\begin{align}
&\frac{|I(X_1^n;Z(X_1^n)|s_0)-I(X_{k+1}^{k+n};Z(X_{k+1}^{k+n}))|}{n}\notag\\
&\hspace{1cm}\le \eps|\mathcal{S}|(\log|\mathcal{X}|+2\log|\mathcal{Z}|)+\frac{2(t(\log|\mathcal{X}|+\log|\mathcal{Z}|)+\log(t+1)-2\eps|\mathcal{S}|\log(|\mathcal{S}|\eps))}{n}.
\end{align}
\end{pr}
\begin{proof}
The proof can be found in Appendix~\ref{blockdifferencestationaryproof}.
\end{proof}
%\begin{remark}
%The bounds in Propositions~\ref{blockdifferenceindependent} and~\ref{blockdifferencestationary} does not depend on the distribution of the input process $X$.
%\end{remark}
One consequence of Proposition~\ref{blockdifferencestationary} is that the mutual information rate does not depend on the initial state.
\begin{corollary}\label{initialindependent}
Let $X$ be a stationary input process and $Z(X)$ be the output obtained by passing $X$ through the channel $W$. Then for any pair of initial states $s_0$ and $s_0^\prime$,
$$
I(X;Z(X)|s_0)=\lim_{n\to\infty}\frac{1}{n}I(X_1^n;Z(X_1^n)|s_0),
$$
exists and $I(X;Z(X)|s_0)=I(X;Z(X)|s_0^\prime).$
\end{corollary}
\begin{proof}
The proof can be found in Appendix~\ref{initialindependentproof}.
\end{proof}
The following theorem justifies that $C_{\mathrm{Shannon}}$ is well-defined and independent of the choice of the initial state $s_0$.
\begin{thm}\label{shannonindependent}
The Shannon capacity $C_{\mathrm{Shannon}}$ does not depend on the initial state of $W_2$ and the limit in the definition of $C_{\mathrm{Shannon}}$ exists.
\end{thm}
\begin{proof}
The proof can be found in Appendix~\ref{shannonindependentproof}.
\end{proof}
The following lemma is from~\cite{dubrushin67} and will be used in the proof of Theorem~\ref{capacity}. 
\begin{lemma}\label{dobrushinlemma}
Let $X$ and $Y$ be finite-valued discrete random variables taking values in $\mathcal{X}$ and $\mathcal{Y}$, respectively. Let $\mathcal{V}$ be a finite set and $\phi: \mathcal{Y}\to\mathcal{V}$ be a function. Let $g_{\phi}(v)=|\{y:\phi(y)=v\}|$. Then
\begin{align}\label{exp}
\EE[|\iota_{X,Y}(X,Y)-\iota_{X,\phi(Y)}(X,\phi(Y))|]\le \max_{v\in\mathcal{V}}\log g_{\phi}(v)
\end{align}
and
\begin{align}\label{prob}
\Pr(\iota_{X,Y}(X,Y)\not=\iota_{X,\phi(Y)}(X,\phi(Y)))\le \Pr(g_{\phi}(Y)\not=1).
\end{align}
\end{lemma}
The following theorem says that for the indecomposable FSC $W_2$ with a stationary and ergodic input process $X$, the asymptotic equipartition property holds.  
\begin{thm}\label{smbfsc}
Let $X=\{X_n\}$ be a stationary and ergodic input process and $Y=\{Y_n\}$ be the output process obtained by passing $X$ through the indecomposable FSC $W_2$. Then 
\begin{align}
-\frac{\log p_{Y_1^n}(Y_1^n)}{n}\to H(Y)\quad \mbox{a.s.\ and in $L^1$}
\end{align}
and 
\begin{align}
-\frac{\log W_2^{n}(Y_1^n|X_1^n,s_0)}{n}\to H(Y|X)\quad \mbox{a.s.\ and in $L^1$}.
\end{align}
\end{thm}
\begin{proof}
The proof can be found in Appendix~\ref{smbfscproof}.
\end{proof}

\section{Proof of Theorem~\ref{capacity} and Corollary~\ref{polar}}\label{proof}

\subsection{Proof of Theorem~\ref{capacity} }\begin{proof}
As $C_{\mathrm{S}}\le C_{\mathrm{Shannon}}$ and $\lim_{m\to\infty}C_{\mathrm{Markov}}^{(m)}\le C_\mathrm{S}$ can be proved easily from the definitions of $C_{\mathrm{S}}$, $C_{\mathrm{Shannon}}$, and $C_{\mathrm{Markov}}^{(m)}$, to prove Theorem~\ref{capacity} it suffices to show $C_{\mathrm{Shannon}}\le C$, $C\le C_{\mathrm{Shannon}}$, $C_{\mathrm{Shannon}}\le C_\mathrm{S}$ and $C_\mathrm{S}\le \lim_{m\to\infty}C_{\mathrm{Markov}}^{(m)}$.

{\bf Proof of $C_{\mathrm{Shannon}}\le C$.} We will adapt Dobrushin's approach~\cite{dubrushin67} to prove the achievability of $C_{\mathrm{Shannon}}.$ Let $\{p_n(\cdot)\}$ be a sequence of probability mass functions and let $\{X_1^n\}$ be a sequence of random vectors in which each $X_1^n$ is distributed according to $p_{n}(\cdot)$ and such that
\begin{align}\label{approachShannon}
\frac{I(X_1^n;Z(X_1^n))}{n}\to C_{\mathrm{Shannon}} \quad \mbox{as $n\to\infty$}.
\end{align} 

For a fixed integer $n$, let $X=\{X_i\}$ be a block independent process with $p_{{X}_{1}^{kn}}(x_1^{kn})=\prod_{i=1}^{k}p_{n}(x_{(i-1)n+1}^{in})$ for any $k$ and $x_1^{kn}\in\mathcal{X}^{kn}$. For notational convenience let $\xi_{i}=X_{(i-1)n+1}^{in}$ and $\xi=\{\xi_i\}$. Then $\xi$ is an i.i.d.\ random process with $p_{\xi_i}(\cdot)=p_{X_1^n}(\cdot)$. Let $\eta_i=Z(X_{(i-1)n+1}^{in})$ be the output obtained when passing $\xi_i=X_{(i-1)n+1}^{in}$ through the channel $W$. Let $\xi_{n,k}=(\xi_1,\cdots,\xi_k)$, ${\eta}_{n,k}=(\eta_1\cdots \eta_{k})$, and $\hat{\eta}_{n,k}=(\eta_1,\cdots,\eta_k)$. Intuitively $\hat{\eta}_{n,k}$ is the output of $W$ with the side information about the number of deletions that occurred during each transmission of $\xi_i$. Let $\hat{\eta}=({\eta}_{1},\eta_2,\cdots)$. We first show that for any $\eps>0$ and for any fixed integer $n$, we have 
\begin{align}\label{auxilary1}
\limsup_{k\to\infty}\Pr\left(\left|\frac{\iota_{\xi_{n,k},\hat{\eta}_{n,k}}}{k}-I(\xi;\hat{\eta})\right|\ge n\eps\right)=0,
\end{align}
where $I(\xi;\hat{\eta})=\lim_{k\to\infty}\frac{I(\xi_{n,k};\hat{\eta}_{n,k})}{k}$.
Using similar ideas as in~\cite{dubrushin67}, we then show that there exists an increasing sequence of integers $\{k_n\}$ such that
\begin{align}\label{auxilary2}
\lim_{n\to\infty}\Pr\left(\left|\frac{\iota_{\xi_{n,k_n},{\eta}_{n,k_n}}}{k_n}-nC_{\mathrm{Shannon}}\right|\ge n\eps\right)=0.
\end{align}
It then follows from~\cite[Theorem 1, pp.~340]{dubrushin59} that $C_{\mathrm{Shannon}}$ is achievable and hence $$ C_{\mathrm{Shannon}}\le C.$$ 

Now we prove~(\ref{auxilary1}). We now define a new channel $\hat{W}$ formed from the concatenation of ${W}_1$ and the FSC ${W}_2$ with the side information about the number of deletions.  Let $\emptyset$ be the empty vector. Let $\zeta_1^{k}=(\zeta_1,\cdots,\zeta_k)$, $s_1^{*k}=(s_1^*,\cdots,s_k^*)$, $z_1^{*k}=(z_1^*,\cdots,z_k^*)$ be respectively $k$-dimensional vectors such that $\zeta_{i}\in\mathcal{X}^n$, $s_i^*\in \mathcal{Z}^{*n}-\{\emptyset\}$, and $z_i^*\in \mathcal{Z}^{*n}$ for $1\le i\le k$. The $j^{\mathrm{th}}$-element of $x_i^*$, $s_i^*$, and $z_i^*$ are denoted by $x_{i,j}$, $s_{i,j}$, and $z_{i,j}$, respectively. Then the probability of obtaining outputs $z_1^{*k}$ and states $s_1^{*k}$ by passing $\zeta_1^{k}$ through the channel $\hat{W}^k$ with the initial state $s_0^*$ is given as 
$$
\hat{W}^{k}(z_1^{*k},s_1^{*k}|\zeta_1^{k},s_0^*)=\prod_{i=1}^{k}p(z^{*}_{i},s_{i}^{*}|\zeta_{i},s_{i-1}^*),
$$
where  
$$
p(z^{*}_{i},s_{i}^{*}|\zeta_{i},s_{i-1}^*)=\begin{cases}\sum_{x^{*}\in\mathcal{X}^{*n}:\ell(x^{*})=
\ell(z_i^*)}W_1^n(x^{*}|\zeta_{i})W_2^{\ell(z_i^{*})}(z_i^{*},s_{i}^{*}|x^{*},s_{i-1,\ell(s_{i-1}^{*})}^{})&\mbox{$\ell(z^*_{i})=\ell(s^*_{i})$}\\
0&\mbox{$\ell(z^*_{i})\not=\ell(s^*_{i})$}.
\end{cases}$$
Intuitively, $\hat{W}^k$ is the channel $W^{nk}$ with the side information about the number of deletions in each transmission of $\zeta_i$. Let $\hat{W}=\{\hat{W}^k\}_{k=1}^{\infty}$ and the initial state of $\hat{W}^k$ is $ s_0^*=s_0$. Then $\hat{W}$ is an indecomposable FSC. (The proof that $\hat{W}$ is indecomposable can be found in Appendix~\ref{indecomposablenew}.) Let $\xi$ be passed through the channel $\hat{W}$. Then the output is $\hat{\eta}$. As $\xi$ is an i.i.d.\ process and $\hat{W}$ is an indecomposable FSC, from Theorem~\ref{smbfsc} we have that for any fixed $n$,
\begin{align}
&\lim_{k\to\infty}-\frac{1}{k}\log p_{\hat{\eta}_{n,k}}(\hat{\eta}_{n,k})=H(\hat{\eta})\quad \mbox{a.s. and in $L^1$}\label{im8}
\end{align}
and
\begin{align}
&\lim_{k\to\infty}-\frac{1}{k}\log \hat{W}^{k}(\hat{\eta}_1^k|\xi_1^k)=H(\hat{\eta}|\xi)\quad\mbox{a.s. and in $L^1$}\label{im9},
\end{align}
respectively. Combining~(\ref{im8}) and~(\ref{im9}), we have~(\ref{auxilary1}), as deisred.
%(\textcolor{red} {needed to be polished})

%As the receiver has side information about the length of the output sequence $Z(X^n(i))$, 

Now we prove~(\ref{auxilary2}). Let $\Omega_{k,n}=\{({z}_1^{*},{z}_2^{*},\cdots,{z}_k^{*}):z^{*}_i\in\mathcal{Z}^{*n} \}$ and let $\phi: \Omega_{k,n}\to \mathcal{Z}^{*kn}$ be the function that maps $({z}_1^{*},{z}_2^{*},\cdots,{z}_k^{*})$ to the concatenation ${z}_1^{*}{z}_2^{*}\cdots{z}_k^{*}$. For any ${z}^{*kn}\in \mathcal{Z}^{*kn}$, let $g_{\phi}(\bar{y}^{*kn})$ be the number of vectors $({z}_1^{*},{z}_2^{*},\cdots,{z}_k^{*})\in \Omega_{k,n}$ such that $\phi({z}_1^{*},{z}_2^{*},\cdots,{z}_k^{*})={z}^{*kn}$.  One easily checks that 
$$
g_{\phi}(z^{*kn})\le {\ell(z^{*kn})+k\choose k}\le {nk+k\choose k}.
$$
As $\phi(\hat{\eta}_{n,k})=\eta_{n,k}$, then from Lemma~\ref{dobrushinlemma} it follows that
\begin{align}\label{expdifference}
\EE[|\iota_{\xi_{n,k},\eta_{n,k}}-\iota_{\xi_{n,k},\hat{\eta}_{n,k}}|]&\le {\log {nk+k\choose k}}\stackrel{(a)}{\le} (n+1)k h\left(\frac{1}{n+1}\right),
\end{align}
where in~$(a)$ we have used the fact that ${(n+1)k\choose k}\le 2^{(n+1)k h\left(\frac{1}{n+1}\right)}$ and $h(x)=-x\log x-(1-x)\log (1-x)$ is the binary entropy function. Using Markov's inequality, we have that
\begin{align}\label{smboriginal}
%&\Pr(|\iota_{\xi_{n,k},\eta_{n,k}}-kn C_{\mathrm{Shannon}}|\ge nk\eps)\notag\\
& \Pr(|\iota_{\xi_{n,k},\eta_{n,k}}-kI(X_1^n;Z(X_1^n))|\ge nk\eps/2)\notag\\
%&\le \Pr(|kI(X_1^n;Z(X_1^n))-kn C_{\mathrm{Shannon}}|\ge nk\eps/2)\notag\\
& \le \Pr(|\iota_{\xi_{n,k},\eta_{n,k}}-\iota_{\xi_{n,k},\hat{\eta}_{n,k}}|\ge nk\eps/4)+\Pr(|\iota_{\xi_{n,k},\hat{\eta}_{n,k}}-kI(X_1^n;Z(X_1^n))|\ge nk\eps/4)\notag\\
&\le \frac{4h(\frac{1}{n+1})}{\eps}+\Pr(|\iota_{\xi_{n,k},\hat{\eta}_{n,k}}-kI(X_1^n;Z(X_1^n))|\ge nk\eps/4)\notag\\
&\le \frac{4h(\frac{1}{n+1})}{\eps}+ \Pr((|\iota_{\xi_{n,k},\hat{\eta}_{n,k}}-kI(\xi;\hat{\eta}))|\ge nk\eps/8)\notag\\
&\hspace{6.5cm}+\Pr(|kI(\xi;\hat{\eta})-kI(X_1^n;Z(X_1^n))|\ge nk\eps/8).
\end{align}
From~(\ref{auxilary1}) it follows that for any fixed integer $n$,
\begin{align}\label{blockchannelsmb}
\lim_{k\to\infty}\Pr\left(\left|\frac{\iota_{\xi_{n,k},\hat{\eta}_{n,k}}}{k}-I(\xi;\hat{\eta})\right|\ge \frac{\eps}{8} \right)=0.
\end{align}
%Note that
%\begin{align}
%I(\xi;\hat{\eta})&=\lim_{k\to\infty}\frac{1}{k}I(X_1^{kn};Z(X_1^n),\cdots,Z(X_{(k-1)n+1}^{kn}))\notag\\
%&=\lim_{k\to\infty}\frac{1}{k}\sum_{i=1}^{k}I(X_{(i-1)n+1}^{in};Z(X_1^n),\cdots,Z(X_{(k-1)n+1}^{kn})|X_1^{(i-1)n})\notag\\
%&=\lim_{k\to\infty}\frac{1}{k}\sum_{i=1}^{k}I(X_{(i-1)n+1}^{in};Z(X_1^n),\cdots,Z(X_{(k-1)n+1}^{kn})|X_1^{(i-1)n})
%\end{align}
From part (iv) in Corollary~\ref{existenceindependent}, we have that for $\eps>0$, there exists an integer $N$ such that when $n\ge N$, 
{\small\begin{align}\label{im111}
\left|I(\xi;\hat{\eta})-\frac{I(X_1^n;Z(X_1^n))}{n}\right|\le \frac{\eps}{10}.
\end{align}}
Combining~(\ref{smboriginal}),~(\ref{blockchannelsmb}) and~(\ref{im111}), we have that for $n\ge N$,~(\ref{auxilary1}) holds.
As $p_{n}(\cdot)$ is chosen such that~(\ref{approachShannon}) holds, we conclude that there exists an increasing sequence of integers $\{k_n\}$ such that~(\ref{auxilary2}) holds, as desired.
%\begin{align}
%\lim_{n\to\infty}\Pr(|\iota_{\xi_{n,k_n},\eta_{n,k_n}}-k_nC_{\mathrm{Shannon}}|\ge nk_n\eps)=0,
%\end{align}
%and
%\begin{align}
%\lim_{n\to\infty}P\left(\left|\frac{1}{k_n}\log \frac{W({\eta}_{n,k_n}|\xi_{n,k_n})}{\Pr({\eta}_{n,k_n})}-nC_{\mathrm{Shannon}}\right|\ge n\eps\right)=0,
%\end{align}

{\bf Proof of $C\le C_{\mathrm{Shannon}}$.} This inequality can be derived by going through the usual converse part using Fano's inequality; for details see~\cite[Section 7.9]{Cover}.

{\bf Proof of $C_{\mathrm{S}}\ge C_{\mathrm{Shannon}}$.} Let $\eps>0$ and let the probability mass function $p_{X_1^n}(\cdot)$ be such that 
$$
\frac{I(X_1^n;Z(X_1^n)|s_0)}{n}\ge \sup_{p_{X_1^n}(\cdot)}\frac{I(X_1^n;Z(X_1^n)|s_0)}{n}-\eps,
$$
where $Z(X_1^n)$ is the output of $W^n$ obtained by passing the input $X_1^n$ through the channel $W$. Using similar ideas as in~\cite{Feinstein}, we construct the stationary and ergodic input process $X$ as follows:
\begin{itemize}
\item[(i)] Construct the block independent process $\hat{X}$ with $p_{\hat{X}_{1}^{kn}}(x_1^{kn})=\prod_{i=1}^{k}p_{X_1^n}(x_{(i-1)n+1}^{in})$ for any $k\ge 1$ and $x_1^{kn}$.
\item[(ii)] Let $V$ be a random variable that is uniformly distributed over $\{0,1,\cdots, n-1\}$. Let $\bar{X}_{k}=\hat{X}_{k+V}$ for any $k\ge 1$. Then one can verify that $\bar{X}=\{\bar{X}_k\}$ is a stationary and ergodic process. 
\end{itemize}
As $\bar{X}$ is stationary, it follows from Corollary~\ref{initialindependent} that for any initial state $s_0$, 
$$\lim_{m\to\infty}\frac{1}{m}I(\bar{X}_1^{m};Z(\bar{X}_1^m)|s_0)\quad
\mbox{exists}.$$  
Here, we note that $\bar{X}_1^m$ refers to the first $m$ random variables in the process $\bar{X}$. Suppose
\begin{align}\label{stationaryshannon}
\lim_{m\to\infty}\frac{1}{m}I(\bar{X}_1^{m};Z(\bar{X}_1^m)|s_0)\ge \frac{I(X_1^n;Z(X_1^n)|s_0)}{n}-2\eps,
\end{align}
then we have that
$$
C_\mathrm{S}\ge \lim_{m\to\infty}\frac{1}{m}I(\bar{X}_1^{m};Z(\bar{X}_1^m)|s_0)\ge C_{\mathrm{Shannon}}-3\eps.
$$
Due to the arbitrariness of $\eps$, we have $C_\mathrm{S}\ge C_{\mathrm{Shannon}}$. Hence to complete the proof, it suffices to show~(\ref{stationaryshannon}). For $0\le j\le n-1$ and $i\ge 1$, let $\hat{X}^{(j)}_i\triangleq\hat{X}_{j+i}$ and $\hat{X}^{(j)}=\{\hat{X}^{(j)}_i\}_{i=1}^{\infty}$. Let $Z(\hat{X}^{(j)})$ be the output obtained by passing $\hat{X}^{(j)}$ through the channel $W$. Then
from~\cite[Lemma 3]{Feinstein} and part (iii) in Corollary~\ref{existenceindependent} it follows that 
$$
\lim_{m\to\infty}\frac{1}{m}I(\bar{X}_1^{m};Z(\bar{X}_1^m)|s_0)=\sum_{i=0}^{n-1}\frac{1}{n}I(\hat{X}^{(j)};Z(\hat{X}^{(j)})|s_0)=I(\hat{X};Z(\hat{X})|s_0).
$$
Thus to complete the proof it suffices to show that
\begin{align}\label{blockshannon}
I(\hat{X};Z(\hat{X})|s_0)=\lim_{m\to\infty}\frac{1}{m}I(\hat{X}_1^{m};Z(\hat{X}_1^m)|s_0)\ge \frac{I(X_1^n;Z(X_1^n)|s_0)}{n}-2\eps.
\end{align}
From Lemma~\ref{sideinformation} it follows that 
{\small\begin{align}
I(\hat{X}_1^{kn};Z(\hat{X}_1^{kn})|s_0)
&\ge I(\hat{X}_1^{kn};Z(\hat{X}_1^{n}),\cdots,Z(\hat{X}_{(k-1)n+1}^{kn}))-k\log (n+1).
\end{align}}By part (iv) in Corollary~\ref{existenceindependent}, we obtain that there exists an integer $N(|\mathcal{X}|,|\mathcal{Z}|,|\mathcal{S}|)$ such that for $n\ge N(|\mathcal{X}|,|\mathcal{Z}|,|\mathcal{S}|)$,
\begin{align}
I(\hat{X};Z(\hat{X})|s_0)&\ge \lim_{k\to\infty}\frac{1}{kn}I(\hat{X}_1^{kn};Z(\hat{X}_1^{n}),\cdots,Z(\hat{X}_{(k-1)n+1}^{kn})|s_0)-\frac{\log (n+1)}{n}\notag\\
&\ge I(X_1^n;Z(X_1^{n})|s_0)-\frac{\log (n+1)}{n}-\eps.
\end{align}
Choosing $n\ge N(|\mathcal{X}|,|\mathcal{Z}|,|\mathcal{S}|)$ such that $\frac{\log (n+1)}{n}\le \eps$, we then have~(\ref{blockshannon})
as desired.

% \eps|\mathcal{S}|(\log|\mathcal{X}|+2\log|\mathcal{Z}|)+\frac{2(t(\log|\mathcal{X}|+\log|\mathcal{Z}|)+\log(t+1)-2\eps|\mathcal{S}|\log(|\mathcal{S}|\eps))}{n}$

{\bf Proof of $C_\mathrm{S}\le \lim_{m\to\infty}C_{\mathrm{Markov}}^{(m)}$.} The idea of the proof is similar to that in~\cite{chensiegel}. Let $\eps>0$ and let $X$ be a stationary and ergodic process such that $I(X;Z(X))\ge C_\mathrm{S}-\eps$.  Let $\delta>0$ be such that 
{\small\begin{align}\label{choice1}
\delta|\mathcal{S}|(\log|\mathcal{X}|+2\log|\mathcal{Z}|)\le \frac{\eps}{2}.
\end{align}}
Let $t$ and $n$ be such that 
{\small\begin{align}\label{choice2}
\frac{\log (n+1)}{n}\le \eps,\quad \frac{H(X_1^n|Z(X_1^n))}{n}\le H(X|Z(X))+\eps,
\end{align}}
and
{\small\begin{align}\label{choice3}
 \quad \frac{2t(\log|\mathcal{X}|+\log|\mathcal{Z}|)+\log(t+1)-4\delta|\mathcal{S}|\log(|\mathcal{S}|\delta)}{n}\le \frac{\eps}{2}.
\end{align} }Let $\hat{X}$ be an $(n-1)^\mathrm{st}$-order Markov process such that $p_{\hat{X}_1^n}(\cdot)=p_{X_1^n}(\cdot)$. We show that
\begin{align}\label{Markov}
I(\hat{X};Z(\hat{X}))\ge I(X;Z(X))-3\eps.
\end{align}
From~(\ref{Markov}) it follows that $C_{\mathrm{Markov}}^{(n-1)}\ge C_\mathrm{S}-4\eps$, which, together with the arbitrariness of $\eps$ and the fact that $C_{\mathrm{Markov}}^{(n-1)}\le \lim_{m\to \infty}C_{\mathrm{Markov}}^{(m)}$, further implies that 
$$
C_\mathrm{S}\le \lim_{m\to \infty}C_{\mathrm{Markov}}^{(m)}.
$$
Thus to complete the proof it suffices to show~(\ref{Markov}).
As $\hat{X}$ is an $(n-1)^\mathrm{st}$-order Markov process, we have that $H(\hat{X})=H(\hat{X}_0|\hat{X}_{-n+1}^{-1})\ge H(X)$. Hence to show~(\ref{Markov}), we only need to prove that
{\small\begin{align}\label{conditionalMarkov}
H(\hat{X}|Z(\hat{X}))=\lim_{k\to\infty} \frac{H(\hat{X}_1^{kn}|Z(\hat{X}_1^{kn}))}{kn}\le \lim_{k\to\infty}\frac{H(X_1^{kn}|Z({X}_1^{kn}))}{kn}+3\eps=H(X|Z(X))+3\eps.
\end{align}}
From Proposition~\ref{blockdifferencestationary}, we have that
\begin{align}\label{im200}
&\frac{|I(\hat{X}_{(i-1)n+1}^{in};Z(\hat{X}_{(i-1)n+1}^{in}))-I(\hat{X}_{1}^{n};Z(\hat{X}_1^{n})|s_0)|}{n}\notag\\
&\le \delta|\mathcal{S}|(\log|\mathcal{X}|+2\log|\mathcal{Z}|)+\frac{2(t(\log|\mathcal{X}|+\log|\mathcal{Z}|)+\log(t+1)-2\delta|\mathcal{S}|\log(|\mathcal{S}|\delta))}{n}\notag\\
&\le \eps.
\end{align}
where~(\ref{im200}) follows from~(\ref{choice1}) and~(\ref{choice3}). 
Then it follows that
\begin{align}\label{markovcond}
H(\hat{X}|Z(\hat{X}))&=\lim_{k\to\infty} \frac{H(\hat{X}_1^{kn}|Z(\hat{X}_1^{kn}))}{kn}\notag\\
&\stackrel{(a)}{=} \lim_{k\to\infty}  \frac{H(\hat{X}_1^{kn}|Z(\hat{X}_1^{n}),\cdots,Z(\hat{X}_{(k-1)n+1}^{kn}))}{kn}+\frac{\log(n+1)}{n}\notag\\
&\stackrel{(b)}{=}\lim_{k\to\infty}  \frac{\sum_{i=1}^{k}H(\hat{X}_{(i-1)n+1}^{in}|\hat{X}_{1}^{(i-1)n},Z(\hat{X}_1^{n}),\cdots,Z(\hat{X}_{(k-1)n+1}^{kn}))}{kn}+\eps\notag\\
&\stackrel{(c)}{\le }\lim_{k\to\infty}  \frac{\sum_{i=1}^{k}H(\hat{X}_{(i-1)n+1}^{in}|Z(\hat{X}_{(i-1)n+1}^{in}))}{kn}+\eps\notag\\
&\stackrel{}{= }\lim_{k\to\infty}  \frac{\sum_{i=1}^{k}(H(\hat{X}_{(i-1)n+1}^{in})-I(\hat{X}_{(i-1)n+1}^{in};Z(\hat{X}_{(i-1)n+1}^{in})))}{kn}+\eps\notag\\
&\stackrel{(d)}{= }\lim_{k\to\infty}  \frac{\sum_{i=1}^{k}(H(\hat{X}_{1}^{n})-I(\hat{X}_{1}^{n};Z(\hat{X}_1^{n})|s_0)+n\eps)}{kn}+\eps\notag\\
&\stackrel{(e)}{=}\frac{H({X}_{1}^{n}|Z({X}_1^{n}))}{n}+2\eps\notag\\
&\le H(X|Z(X))+3\eps,
\end{align}
where~$(a)$ follows from the fact that $$
\frac{H(\hat{X}_1^{kn}|Z(\hat{X}_1^{kn}))}{kn}\le \frac{H(\hat{X}_1^{kn}|Z(\hat{X}_1^{n}),\cdots,Z(\hat{X}_{(k-1)n+1}^{kn}))}{kn}+\frac{\log(n+1)}{n},
$$~$(b)$ follows from~(\ref{choice2}),~$(c)$ follows from that conditioning does not increase entropy,~$(d)$ follows from the stationarity of $X$ and~(\ref{im200}),~$(e)$ follows from the fact that $\hat{X}_{1}^{n}$ and ${X}_{1}^{n}$ have the same distribution.
\end{proof}

\subsection{Proof of Corollary~\ref{polar}}
\begin{proof}
In Theorem~\ref{capacity}, we showed that given any $\eps>0$, there is an $m^{\mathrm{th}}$-order Markov process $\bar{X}$ such that $I(\bar{X};Z(\bar{X}))\ge C-{\eps}/{4}$. Suppose there exists an $m^{\mathrm{th}}$-order irreducible and ergodic Markov process $X$ be such that $I(X;Z(X))\ge I(\bar{X};Z(\bar{X}))-\eps/4$. Let $R_{k}$ be the rate of the coding scheme in~\cite{idodeletion} associated with the input process ${X}$. Then from Theorem~\ref{idopolar}, Tal~et.al obtained that there exists an integer $M(X,\eps)$ such that $R_k\ge I({X};Z({X}))-{\eps}/{2}$ for $k\ge M(X,\eps)$.  Hence we have that for $k\ge M(X,\eps)$,
$$
R_k\ge I({X};Z({X}))-\frac{\eps}{2}\ge I(\bar{X};Z(\bar{X}))-\frac{\eps}{4}-\frac{\eps}{2}\ge C-\eps.
$$
To complete the proof, it suffices to show the existence of an $m^{\mathrm{th}}$-order irreducible and ergodic Markov process ${X}$ satisfying $I({X};Z(X))\ge I(\bar{X};Z(\bar{X}))-{\eps}/{4}$.

 Let $n$ be an integer such that 
 $$
\frac{H(\bar{X}_1^n|Z(\bar{X}_1^n))}{n}\ge H(\bar{X}|Z(\bar{X}))-\frac{\eps}{100}\mbox{,}\quad \frac{\log (n+1)}{n}\le \frac{\eps}{100},
$$
and let $\delta>0$
\begin{align}\label{delta}
3\delta\log|\mathcal{X}|+\frac{2(2t\log|\mathcal{X}|+\log(t+1)-2\delta\log\delta)}{n}\le\frac{\eps}{100}.
\end{align}

We now define the desired $m^{\mathrm{th}}$-order Markov process ${X}$ by 
$$
p_{{X}_1^{m+1}}(x_1^{m+1})=\begin{cases}p_{\bar{X}_{1}^{m+1}}(x_1^{m+1})-\delta_1&p_{\bar{X}_{m+1}|\bar{X}_{1}^{m}}(x_{m+1}|x_1^{m})>0,\\
\delta_2&\mbox{otherwise},
\end{cases}
$$ where $\delta_1$ and $\delta_2$ are chosen as follows. As $\bar{X}$ is an $m^{\mathrm{th}}$-order Markov process, $H(\bar{X}_1^n|Z(\bar{X}_1^n))$ is continuous in $p_{{X}_1^{m+1}}(\cdot)$. Together with the continuity of $H(X_{m+1}|X_1^{m+1})$ in $p_{{X}_1^{m+1}}(\cdot)$, we conclude that there exists $\delta_1$ and $\delta_2$ such that 
\begin{itemize}
\item[(i)] $p_{{X}_1^{m+1}}(x_1^{m+1})>0$ \mbox{for all $x_1^{m+1}$};
\item[(ii)] $\sum_{x_1^{m+1}} p_{{X}_1^{m+1}}(x_1^{m+1})=1$;
\item[(iii)] $|H(X_{m+1}|X_1^{m+1})-H(\bar{X}_{m+1}|\bar{X}_1^{m+1})|\le \frac{\eps}{100}$;
\item[(iv)] $\frac{|H({X}_1^n|Z({X}_1^n))-H(\bar{X}_1^n|Z(\bar{X}_1^n))|}{n}\le \frac{\eps}{100}$.
\end{itemize}
 As $p_{{X}_1^{m+1}}(x_1^{m+1})>0$ for all $x_1^{m+1}$, $X$ is irreducible and aperiodic. From Proposition~\ref{blockdifferencestationary} it follows that for $\delta$ in~(\ref{delta})
 and $n\ge t\ge K(\delta)$,
{\small\begin{align}\label{im100}
\left|\frac{I({X}_{(i-1)n+1}^{in};Z({X}_{(i-1)n+1}^{in}))}{n}-\frac{I({X}_{1}^{n};Z({X}_{1}^{n}))}{n}\right|\le 3\delta\log|\mathcal{X}|+\frac{2(2t\log|\mathcal{X}|+\log(t+1)-2\delta\log\delta)}{n}.
\end{align}}

Then similar to the derivation of~(\ref{markovcond}), we obtain that
{\small\begin{align}
I(X;Z(X))&=\lim_{k\to\infty}\frac{I(X_1^{kn};Z(X_1^{kn}))}{kn} \notag\\
&\stackrel{(a)}{\ge} H(X)-\lim_{k\to\infty}\frac{H({X}_1^{kn}|Z({X}_1^{n}),\cdots,Z({X}_{(k-1)n+1}^{kn}))}{kn}-\frac{\log(n+1)}{n}\notag\\
&\stackrel{}{\ge} H(X)-\lim_{k\to\infty}\frac{\sum_{i=1}^{k}H({X}_{(i-1)n+1}^{in}|Z({X}_1^{n}),\cdots,Z({X}_{(k-1)n+1}^{kn}),{X}^{(i-1)n}_{1})}{kn}-\frac{\eps}{100}\notag\\
&\stackrel{(b)}{\ge} H(X)-\lim_{k\to\infty}\frac{\sum_{i=1}^{k}H({X}_{(i-1)n+1}^{in}|Z({X}_{(i-1)n+1}^{in}))}{kn}-\frac{\eps}{100}\notag\\
&\stackrel{}{=} H(X)-\lim_{k\to\infty}\frac{\sum_{i=1}^{k}H({X}_{(i-1)n+1}^{in})-I({X}_{(i-1)n+1}^{in};Z({X}_{(i-1)n+1}^{in}))}{kn}-\frac{\eps}{100}\notag\\
&\stackrel{(c)}{=} H(X)-\frac{H({X}_{1}^{n})}{n}+\lim_{k\to\infty}\frac{\sum_{i=1}^{k}I({X}_{(i-1)n+1}^{in};Z({X}_{(i-1)n+1}^{in}))}{kn}-\frac{\eps}{100}\notag\\
&\stackrel{(d)}{\ge} H(X)-\frac{H({X}_{1}^{n})}{n}+\lim_{k\to\infty}\frac{\sum_{i=1}^{k}I({X}_{1}^{n};Z({X}_{1}^{n}))}{kn}-\frac{2\eps}{100}\notag\\
&= H(X)-\frac{H({X}_{1}^{n}|Z({X}_1^{n}))}{n}-\frac{2\eps}{100}\notag\\
&\stackrel{(e)}{\ge} H(\bar{X})- \frac{H({{X}}_{1}^{n}|Z({{X}}_1^{n}))}{n}-\frac{3\eps}{100}\notag\\
&\stackrel{(f)}{\ge} H(\bar{X})- \frac{H({\bar{X}}_{1}^{n}|Z({\bar{X}}_1^{n}))}{n}-\frac{4\eps}{100}\notag\\
&\stackrel{}{\ge} H(\bar{X})- H({\bar{X}}|Z({\bar{X}}))-\frac{5\eps}{100}\notag\\
&\stackrel{}{\ge} I(\bar{X};Z(\bar{X}))-\frac{\eps}{4},\notag
\end{align}}
where~$(a)$ follows from the fact that 
{\small$$
H({X}_1^{kn}|Z({X}_1^{kn}))\ge H({X}_1^{kn}|Z({X}_1^{n}),\cdots,Z({X}_{(k-1)n+1}^{kn}))-k\log(n+1),$$}$(b)$ follows from that conditioning does not increase entropy,
~$(c)$ follows from the stationarity of $X$,
~$(d)$ follows from~(\ref{im100}) and the choice of $n$ and $\delta$,
~$(e)$ follows from (iii) in the definition of $p_{{X}_1^{m+1}}(x_1^{m+1})$,
~$(f)$ follows from (iv) in the definition of $p_{{X}_1^{m+1}}(x_1^{m+1})$.
\end{proof}

\section*{Appendices} \appendix
\section{Proof of Lemma~\ref{indecomposable}}\label{indecomposableproof}
\begin{proof}
As the FSC $W_2$ is indecomposable, there exists an integer $K_1$ such that if $m\ge K_1$
\begin{align}\label{fscindecomposable}
\max_{s\in\mathcal{S}}|\Pr(S_m=s|X_1^m=x_1^m,S_0=s_0)-\Pr(S_m=s|X_1^m=x_1^m,S_0=s_0^\prime)|\le \frac{\eps}{2},
\end{align} 
for any $s_0$, $s_0^\prime$, and $y_1^m$.
\end{proof}
Note that $N(k)$ is a binomial random variable with $k$ trials and probability of success $d$. Therefore there exists an integer $K_2$ such that $\Pr(N(k)\le K_1)\le \eps/4$ for all $k\ge  K_2$. 
 Then for $k\ge  K_2$, we have that
\begin{align}
&\Pr(S_{k-N(k)}=s|X_1^k=x_1^k,S_0=s_0)\notag\\
&=\Pr(S_{k-N(k)}=s, N(k)\le K_1|X_1^k=x_1^k,S_0=s_0)\notag\\
&\hspace{3cm}+\sum_{i=k-K_1}^{k}\Pr(N(k)=i,S_{k-i}=s|X_1^k=x_1^k,S_0=s_0)\notag\\
&\stackrel{}{\le}\Pr(N(k)\le K_1|X_1^k=x_1^k,S_0=s_0)+\sum_{i=k-K_1}^{k}\Pr(N(k)=i,S_{k-i}=s|X_1^k=x_1^k,S_0=s_0)\notag\\
&\stackrel{(a)}{\le}\Pr(N(k)\le K_1)+\sum_{i=k-K_1}^{k}\Pr(N(k)=i,S_{k-i}=s|X_1^k=x_1^k,S_0=s_0)\notag\\
&\stackrel{(b)}{\le} \frac{\eps}{4}+\sum_{i=k-K_1}^{k}\Pr(N(k)=i,S_{k-i}=s|X_1^k=x_1^k,S_0=s_0)\notag\\
&\stackrel{(c)}{=}\frac{\eps}{4}+\sum_{i=k-K_1}^{k}\sum_{x_1^{k-i}}\Pr(Y(X_1^{k})=x_1^{k-i}|X_1^k=x_1^k,s_0) \Pr(S_{k-i}=s|Y(X_1^{k})=x_1^{k-i},S_0=s_0),\notag
\end{align}
where~$(a)$ follows from the fact that $N(k)$ is independent of $(X_1^k,S_0)$,~$(b)$ follows from the choice of $K_2$,~$(c)$ follows from the facts that 
$$
\Pr(N(k)=i,Y(X_1^{k})=x_1^{k-i}|X_1^k=x_1^k,s_0)=\Pr(Y(X_1^{k})=x_1^{k-i}|X_1^k=x_1^k,s_0)
$$
and that $S_{k-i}$ is independent of $X_1^k$ given $(Y(X_{1}^k),S_0)$. Then we have that for any current state $s$,
\begin{align}
&|\Pr(S_{k-N(k)}=s|X_1^k=x_1^k,S_0=s_0)-\Pr(S_{k-N(k)}=s|X_1^k=x_1^k,S_0=s_0^\prime)|\notag\\
&\le \frac{\eps}{2}+\sum_{i=k-K_1}^{k}\sum_{x_1^{k-i}}\left\{\Pr(Y(X_1^{k})=x_1^{k-i}|X_1^k=x_1^k,s_0)\right.\notag\\
&\hspace{1cm}\left.| \Pr(S_{k-i}=s|Y(X_1^{k})=x_1^{k-i},S_0=s_0)-\Pr(S_{k-i}=s|Y(X_1^{k})=x_1^{k-i},S_0=s_0^\prime)|\right\}\notag\\
&\stackrel{(a)}{\le} \frac{\eps}{2}+\frac{\eps}{2}\sum_{i=k-K_1}^{k}\sum_{x_1^{k-i}}\Pr(Y(X_1^{k})=x_1^{k-i}|X_1^k=x_1^k,s_0)\notag\\
&\le \eps,\notag
\end{align}
where~$(a)$ follows from~(\ref{fscindecomposable}).

\section{Proof of Lemma~\ref{sideinformation}}\label{sideinformationproof}
\begin{proof}
Note that 
\begin{align}\label{im600}
I(X_1^n;Z(X_1^n)|s_0)&=I(X_1^{n};Z(X_1^n),N(t_1),\cdots,N_{t_{m-1}+1}(t_{m}-t_{m-1})|s_0)\notag\\
&\hspace{1cm}-I(X_1^n;N(t_1),\cdots,N_{t_{m-1}+1}(t_{m}-t_{m-1})|Z(X_1^n),s_0).
\end{align}
As $N_{t_{i-1}+1}(t_{i}-t_{i-1})$ can take at most $t_{i}-t_{i-1}+1$ values for any $1\le i\le m$, we have that
\begin{align}\label{first}
I(X_1^n;N(t_1),\cdots,N_{t_{m-1}+1}(t_{m}-t_{m-1})|Z(X_1^n),s_0)\le \sum_{i=1}^{m}\log(t_{i}-t_{i-1}+1),
\end{align}
which, together with~(\ref{im600}), further implies~(\ref{sideinformationequation}).
\end{proof}
\section{Proof of Proposition~\ref{blockdifferenceindependent}}\label{blockdifferenceindependentproof}
\begin{proof}
Let $\eps>0$ and $K(\eps)$ be given as in Lemma~\ref{indecomposable}. From Lemma~\ref{sideinformation}, it follows that for any integer $k$,
%\begin{align}
%I(X_1^n;Z(X_1^n)|s_0)&=I(X_1^n;N(k),Z(X_1^n)|s_0)-I(X_1^n;N(k)|Z(X_1^n),s_0)\notag\\
%&=I(X_1^n;Z(X_1^k),Z(X_{k+1}^n)|s_0)-I(X_1^n;N(k)|Z(X_1^n),s_0).
%\end{align}
%As $N(k)$ can only take $k+1$ values, we have that
\begin{align}\label{first}
0\le I(X_1^n;Z(X_1^k),Z(X_{k+1}^n)|s_0)-I(X_1^n;Z(X_1^n)|s_0)\le \log (k+1).
\end{align}
Using the chain rule, we obtain that
 \begin{align}\label{firstblock}
 &I(X_1^n;Z(X_1^k),Z(X_{k+1}^n)|s_0)\notag\\
 &=I(X_1^k;Z(X_1^k),Z(X_{k+1}^n)|s_0)+I(X_{k+1}^n;Z(X_1^k)|X_1^k,s_0)+I(X_{k+1}^n;Z(X_{k+1}^n)|Z(X_1^k),X_1^k,s_0)\notag\\
 &\stackrel{(a)}{=}I(X_1^k;Z(X_1^k),Z(X_{k+1}^n)|s_0)+I(X_{k+1}^n;Z(X_{k+1}^n)|Z(X_1^k),X_1^k,s_0),
 \end{align}
 where~$(a)$ follows from the conditional independence of $X_{k+1}^n$ and $Z(X_1^k)$ given $(X_1^k,S_0)$. Then it follows that
 {\small\begin{align}\label{im1}
& |I(X_1^n;Z(X_1^n)|s_0)-I(X_1^n;Z(X_1^n)|s_0^\prime)|\notag\\
&\le 2\log(k+1) +|I(X_1^n;Z(X_1^k),Z(X_{k+1}^n)|s_0)-I(X_1^n;Z(X_1^k),Z(X_{k+1}^n)|s_0^\prime)|\notag\\
&=2\log(k+1) +|I(X_1^k;Z(X_1^k),Z(X_{k+1}^n)|s_0)-I(X_1^k;Z(X_1^k),Z(X_{k+1}^n)|s_0^\prime)|\notag\\
&\hspace{0.5cm}+|I(X_{k+1}^n;Z(X_{k+1}^n)|Z(X_1^k),X_1^k,s_0)-I(X_{k+1}^n;Z(X_{k+1}^n)|Z(X_1^k),X_1^k,s_0^\prime)|\notag\\
&\le 2[\log(k+1) +k\log |\mathcal{X}|]\notag\\
&\hspace{0.5cm}+|I(X_{k+1}^n;Z(X_{k+1}^n)|Z(X_1^k),X_1^k,s_0)-I(X_{k+1}^n;Z(X_{k+1}^n)|Z(X_1^k),X_1^k,s_0^\prime)|.
 \end{align}}
To complete the proof, it suffices to show
\begin{align}
&|I(X_{k+1}^n;Z(X_{k+1}^n)|Z(X_1^k),X_1^k,s_0)-I(X_{k+1}^n;Z(X_{k+1}^n)|Z(X_1^k),X_1^k,s_0^\prime)|\notag\\
& \le 2\log |\mathcal{S}|+(n-k)\log(|\mathcal{X}|\eps).
\end{align}
From Lemma~\ref{conditionallemma}, it follows that 
$$
|I(X_{k+1}^n;Z(X_{k+1}^n)|Z(X_1^k),X_1^k,s_0)-I(X_{k+1}^n;Z(X_{k+1}^n)|S_{k-N(k)},Z(X_1^k),X_1^k,s_0)|\le \log|\mathcal{S}|
$$
and
\begin{align}
|I(X_{k+1}^n;Z(X_{k+1}^n)|Z(X_1^k),X_1^k,s_0^\prime)-I(X_{k+1}^n;Z(X_{k+1}^n)|S_{k-N(k)},Z(X_1^k),X_1^k,s_0^\prime)|\le \log|\mathcal{S}|.\notag
\end{align}
As given $S_{k-N(k)}$ and $(X_1^k,S_0)$, $Z(X_1^k)$ is conditionally independent of $(X_{k+1}^n,Z(X_{k+1}^n))$, we have that
$$
I(X_{k+1}^n;Z(X_{k+1}^n)|S_{k-N(k)},Z(X_1^k),X_1^k,s_0)=I(X_{k+1}^n;Z(X_{k+1}^n)|S_{k-N(k)},X_1^k,s_0).
$$
From the definition of conditional mutual information, it follows that
{\small\begin{align}\label{firstconditionalblock}
&I(X_{k+1}^n;Z(X_{k+1}^n)|S_{k-N(k)},X_1^k,s_0)\notag\\
&=\sum_{s,x_1^k}\Pr(X_1^k=x_1^k)\Pr(S_{k-N(k)}=s|X_1^k=x_1^k,s_0)I(X_{k+1}^n;Z(X_{k+1}^n)|S_{k-N(k)}=s,X_1^k=x_1^k).
\end{align}}
Thus we have that
{\small\begin{align}
&|I(X_{k+1}^n;Z(X_{k+1}^n)|Z(X_1^k),X_1^k,s_0)-I(X_{k+1}^n;Z(X_{k+1}^n)|Z(X_1^k),X_1^k,s_0^\prime)|\notag\\
&\le 2\log |\mathcal{S}|+|\sum_{s,x_1^k}\left\{p_{S_{n-N(n)}}(s^\prime)p_{X_{1}^{k}}(x_1^k)I(X_{k+1}^{n};Z(X_{k+1}^{n})|S_{k-N(k)}=s,X_{1}^{k}=x_1^k)\right.\notag\\
&\hspace{3.5cm}\times\left.|p_{S_{k-N(k)}|X_{1}^{k},S_{0}}(s|x_1^k,s_0^\prime)-p_{S_{k-N(k)}|X_{1}^{k},S_{0}}(s|x_1^k,s_0)|\right\}\notag\\
&\stackrel{(c)}{=}2\log |\mathcal{S}|+I(X_{k+1}^{n};Z(X_{k+1}^{n})|S_{k-N(k)},X_{1}^{k})\eps\notag\\
&\le 2\log |\mathcal{S}|+(n-k)\log(|\mathcal{X}|\eps),
\end{align}}
where~$(c)$ follows from Lemma~\ref{indecomposable}.
Then the proof is complete.
\end{proof}

\section{Proof of Corollary~\ref{existenceindependent}}\label{existenceindependentproof}
\begin{proof}
{\bf Proof of~(i):} Let $\eps>0$ and $K(\eps)$ be given as in Lemma~\ref{indecomposable}. Using Lemma~\ref{sideinformation}, we have that
{\small\begin{align}\label{second}
|I(X_{n+1}^{2n};Z(X_{n+1}^{n+k}),Z(X_{n+k+1}^{2n})|S_{n-N(n)})-I(X_{n+1}^{2n};Z(X_{n+1}^{n+k}),Z(X_{n+k+1}^{2n})|S_{n-N(n)})|\le \log(k+1).
\end{align}}
Using similar argument as in the derivations of~(\ref{first}) and~(\ref{firstblock}),
 \begin{align}\label{secondblock}
 &I(X_{n+1}^{2n};Z(X_{n+1}^{n+k}),Z(X_{n+k+1}^{2n})|S_{n-N(n)})=I(X_{n+1}^{n+k};Z(X_{n+1}^{n+k}),Z(X_{n+k+1}^{2n})|S_{n-N(n)})\notag\\
 &\hspace{6cm}+I(X_{n+k+1}^{2n};Z(X_{n+k+1}^{2n})|Z(X_{n+1}^{n+k}),X_{n+1}^{n+k},S_{n-N(n)}).
 \end{align}
 Then it follows that
 {\small\begin{align}\label{im1}
& |I(X_1^n;Z(X_1^n)|s_0)-I(X_{n+1}^{2n};Z(X_{n+1}^{2n})|S_{n-N(n)})|\notag\\
&=2\log(k+1) +|I(X_1^n;Z(X_1^k),Z(X_{k+1}^n)|s_0)-I(X_{n+1}^{2n};Z(X_{n+1}^{n+k}),Z(X_{n+k+1}^{2n})|S_{n-N(n)})|\notag\\
&=2\log(k+1) +|I(X_1^k;Z(X_1^k),Z(X_{k+1}^n)|s_0)-I(X_{n+1}^{n+k};Z(X_{n+1}^{n+k}),Z(X_{n+k+1}^{2n})|S_{n-N(n)})|\notag\\
&\quad +|I(X_{k+1}^n;Z(X_{k+1}^n)|Z(X_1^k),X_1^k,s_0)-I(X_{n+k+1}^{2n};Z(X_{n+k+1}^{2n})|Z(X_{n+1}^{n+k}),X_{n+1}^{n+k},S_{n-N(n)})\notag\\
&le 2[\log(k+1) +k\log |\mathcal{X}|]\notag\\
&\quad +|I(X_{k+1}^n;Z(X_{k+1}^n)|Z(X_1^k),X_1^k,s_0)-I(X_{n+k+1}^{2n};Z(X_{n+k+1}^{2n})|Z(X_{n+1}^{n+k}),X_{n+1}^{n+k},S_{n-N(n)})|.
 \end{align}}
To complete the proof it suffices to show 
\begin{align}\label{im500}
&|I(X_{k+1}^n;Z(X_{k+1}^n)|Z(X_1^k),X_1^k,s_0)-I(X_{n+k+1}^{2n};Z(X_{n+k+1}^{2n})|Z(X_{n+1}^{n+k}),X_{n+1}^{n+k},S_{n-N(n)})|\notag\\
&\quad\le 2\log |\mathcal{S}|+(n-k)\eps\log|\mathcal{X}|.
\end{align}
From~Lemma~\ref{conditionallemma}, it follows that 
$$
|I(X_{k+1}^n;Z(X_{k+1}^n)|Z(X_1^k),X_1^k,s_0)-I(X_{k+1}^n;Z(X_{k+1}^n)|S_{k-N(k)},Z(X_1^k),X_1^k,s_0)|\le \log|\mathcal{S}|
$$
and
\begin{align}
&|I(X_{n+k+1}^{2n};Z(X_{n+k+1}^{2n})|S_{k-N_{n+1}(k)+n-N(n)},Z(X_{n+1}^{n+k}),X_{n+1}^{n+k},S_{n-N(n)})\notag\\
&\hspace{3cm}-I(X_{n+k+1}^{2n};Z(X_{n+k+1}^{2n})|Z(X_{n+1}^{n+k}),X_{n+1}^{n+k},S_{n-N(n)})|\le \log|\mathcal{S}|\notag.
\end{align}
As {\small\begin{align}
&p_{Z(X_{n+1}^{n+k}),(X_{n+k+1}^{2n},Z(X_{n+k+1}^{2n}))|S_{k-N_{n+1}(k)+n-N(n)}, S_{n-N(n)},X_{n+1}^{n+k}}(\cdot|\cdot)\notag\\
&=p_{Z(X_{n+1}^{n+k})|S_{k-N_{n+1}(k)+n-N(n)}, S_{n-N(n)},X_{n+1}^{n+k}}(\cdot|\cdot)p_{X_{n+k+1}^{2n},Z(X_{n+k+1}^{2n})|S_{k-N_{n+1}(k)+n-N(n)}, S_{n-N(n)},X_{n+1}^{n+k}}(\cdot|\cdot),\notag
\end{align} }we have that
{\small
\begin{align}
&I(X_{n+k+1}^{2n};Z(X_{n+k+1}^{2n})|S_{k-N_{n+1}(k)+n-N(n)},Z(X_{n+1}^{n+k}),X_{n+1}^{n+k},S_{n-N(n)})\notag\\
&\hspace{4cm}=I(X_{n+k+1}^{2n};Z(X_{n+k+1}^{2n})|S_{k-N_{n+1}(k)+n-N(n)},X_{n+1}^{n+k},S_{n-N(n)})).
\end{align}}
From the definition of conditional mutual information, it follows that
\begin{align}\label{secondconditionalblock}
&I(X_{n+k+1}^{2n};Z(X_{n+k+1}^{2n})|S_{k-N_{n+1}(k)+n-N(n)},X_{n+1}^{n+k},S_{n-N(n)}))\notag\\
&\stackrel{(a)}{=}\sum_{s,s^\prime,x_1^k}\bigg\{p_{S_{n-N(n)}}(s^\prime)p_{X_{n+1}^{n+k}}(x_1^k)p_{S_{k-N_{n+1}(k)+n-N(n)}|X_{n+1}^{n+k},S_{n-N(n)}}(s|x_1^k,s^\prime)\notag\\
&\hspace{2cm}\times I(X_{n+k+1}^{2n};Z(X_{n+k+1}^{2n})|S_{k-N_{n+1}(k)+n-N(n)}=s,X_{n+1}^{n+k}=x_1^k)\bigg\}\notag\\
&\stackrel{(b)}{=}\sum_{s,s^\prime,x_1^k}p_{S_{n-N(n)}}(s^\prime)p_{X_{1}^{k}}(x_1^k)p_{S_{k-N(k)}|X_{1}^{k},S_{0}}(s|x_1^k,s^\prime)I(X_{k+1}^{n};Z(X_{k+1}^{n})|S_{0}=s,X_{1}^{k}=x_1^k),
\end{align}
where~$(a)$ follows from the independence of $S_{n-N(n)}$ and $X_{n+1}^{n+k}$ and $(b)$ follows from the fact that $X_1^n$ and $X_{n+1}^{2n}$ are independent and have the same distribution.
It then follows that
{\small\begin{align}\label{im501}
&|I(X_{k+1}^n;Z(X_{k+1}^n)|Z(X_1^k),X_1^k,s_0)-I(X_{n+k+1}^{2n};Z(X_{n+k+1}^{2n})|Z(X_{n+1}^{n+k}),X_{n+1}^{n+k},S_{n-N(n)})|\notag\\
&\le 2\log |\mathcal{S}|+\sum_{s,s^\prime,x_1^k}\bigg\{p_{S_{n-N(n)}}(s^\prime)p_{X_{1}^{k}}(x_1^k)I(X_{k+1}^{n};Z(X_{k+1}^{n})|S_{k-N(k)}=s,X_{1}^{k}=x_1^k)\notag\\
&\hspace{3.5cm}\times|p_{S_{k-N(k)}|X_{1}^{k},S_{0}}(s|x_1^k,s^\prime)-p_{S_{k-N(k)}|X_{1}^{k},S_{0}}(s|x_1^k,s_0)|\bigg\}\notag\\
&\stackrel{(c)}{=}2\log |\mathcal{S}|+I(X_{k+1}^{n};Z(X_{k+1}^{n})|S_{k-N(k)},X_{1}^{k})\eps,
\end{align}}
where~$(c)$ follows from Lemma~\ref{indecomposable}. Then~(\ref{im500}) follows from~(\ref{im501}), as desired.

{\bf Proof of~(ii):} We first show that $\lim_{k\to\infty}I(X_1^{kn};Z(X_1^{kn})|s_0)/k$ exists. 

Let $a_{k}=I(X_1^{kn};Z(X_1^{kn})|s_0)/k$ and let $k_1$ and $k_2$ be two positive integers. Then from Lemma~\ref{sideinformation}, we have that 
{\small\begin{align}
|I(X_1^{(k_1+k_2)n};Z(X_1^{(k_1+k_2)n})|s_0)-I(X_1^{(k_1+k_2)n};Z(X_1^{k_1n}),Z(X_{k_1n+1}^{(k_1+k_2)n})|s_0)|\le \log(1+k_1n)
\end{align}}
and
{\small\begin{align}
&I(X_1^{(k_1+k_2)n};Z(X_1^{k_1n}),Z(X_{k_1n+1}^{(k_1+k_2)n})|s_0)\notag\\
&=I(X_1^{k_1n};Z(X_1^{k_1n}),Z(X_{k_1n+1}^{(k_1+k_2)n})|s_0)+I(X_{k_1n+1}^{(k_1+k_2)n};Z(X_1^{k_1n}),Z(X_{k_1n+1}^{(k_1+k_2)n})|X_1^{k_1n},s_0)\notag\\
&\ge I(X_1^{k_1n};Z(X_1^{k_1n})|s_0)+I(X_{k_1n+1}^{(k_1+k_2)n};Z(X_{k_1n+1}^{(k_1+k_2)n})|X_1^{k_1n},s_0)\\
&=k_1a_{k_1}+I(X_{k_1n+1}^{(k_1+k_2)n};Z(X_{k_1n+1}^{(k_1+k_2)n})|X_1^{k_1n},s_0)\notag\\
&=k_1a_{k_1}+I(X_{k_1n+1}^{(k_1+k_2)n};Z(X_{k_1n+1}^{(k_1+k_2)n})|S_{{k_1n}-N(k_1n)})-\log |\mathcal{S}|.
\end{align}}
For any $s_0$, we have $p_{X_{k_1n+1}^{(k_1+k_2)n})|S_{{k_1n}-N(k_1n)}}(\cdot|s_0)=p_{X_{k_1n+1}^{(k_1+k_2)n}}(\cdot)$. Together with the fact that
 $p_{X_{k_1n+1}^{(k_1+k_2)n}}(\cdot)=p_{X_{1}^{k_2n}}(\cdot)$, we have that for any $s_0$,
$$
p_{X_{1}^{k_2n},Z(X_{1}^{k_2n})|S_0}(\cdot|s_0)=p_{X_{1}^{(k_1+k_2)n},Z(X_{k_1n+1}^{(k_1+k_2)n})|S_{{k_1n}-N(k_1n)}}(\cdot|s_0),
$$
which further implies that
{\small\begin{align}
&I(X_{k_1n+1}^{(k_1+k_2)n};Z(X_{k_1n+1}^{(k_1+k_2)n})|S_{{k_1n}-N(k_1n)})\notag\\
&=\sum_{s}p_{S_{{k_1n}-N(k_1n)}}(s)I(X_{k_1n+1}^{(k_1+k_2)n};Z(X_{k_1n+1}^{(k_1+k_2)n})|S_{{k_1n}-N(k_1n)}=s)\notag\\
&=\sum_{s}p_{S_{{k_1n}-N(k_1n)}}(s)I(X_{1}^{k_2n};Z(X_{1}^{k_2n})|S_0=s).
\end{align}}
Let $\eps>0$. Then it follows from Proposition~\ref{blockdifferenceindependent} that for $K(\eps)\le k\le k_2n$,
\begin{align}
&|I(X_{k_1n+1}^{(k_1+k_2)n};Z(X_{k_1n+1}^{(k_1+k_2)n})|S_{{k_1n}-N(k_1n)})-I(X_{1}^{k_2n};Z(X_{1}^{k_2n})|s_0)|\notag\\
&\le \sum_{s}p_{S_{{k_1n}-N(k_1n)}}(s)|I(X_{1}^{k_2n};Z(X_{1}^{k_2n})|S_0=s)-I(X_{1}^{k_2n};Z(X_{1}^{k_2n})|s_0)|\notag\\
&\le 2(\log|\mathcal{S}|+\log(k+1)+k\log|\mathcal{X}|)+(k_2n-k)\eps\log|\mathcal{X}|,
\end{align}
which further implies that
{\small\begin{align}
&I(X_1^{(k_1+k_2)n};Z(X_1^{(k_1+k_2)n})|s_0)\notag\\
& \ge I(X_1^{(k_1+k_2)n};Z(X_1^{k_1n}),Z(X_{k_1n+1}^{(k_1+k_2)n})|s_0)-\log(1+k_1n)\notag\\
&\ge k_1a_{k_1}+I(X_{k_1n+1}^{(k_1+k_2)n};Z(X_{k_1n+1}^{(k_1+k_2)n})|S_{{k_1n}-N(k_1n)})-\log(1+k_1n)-\log |\mathcal{S}|\notag\\
&\ge k_1a_{k_1}+k_2a_{k_2}-2(\log|\mathcal{S}|+\log(k+1)-k\log|\mathcal{X}|)+(k_2n-k)\eps\log|\mathcal{X}|-\log( |\mathcal{S}|(1+k_1n)).
\end{align}}
Fix $k=K(\eps)+1$ and choose $k_1>k_2$ such that  
$$
\frac{-2(\log|\mathcal{S}|+\log(k+1)+k\log|\mathcal{X}|)-(k_2n-k)\eps\log|\mathcal{X}|-\log( |\mathcal{S}|(1+k_1n))}{k_1+k_2}\ge -\eps.
$$
Then
$(k_1+k_2)(a_{k_1+k_2}-\eps)\ge k_1(a_{k_1}-\eps)+k_2(a_{k_2}-\eps)$. From~\cite[Lemma 2, pp.~112]{gallagerbook}, $\lim_{k\to \infty}a_k=\lim_{n\to \infty}\left\{a_k-\eps/k\right\}$ exists.

For any integers $m$ and $n$, let $m=kn+r(m,n)$ and $r(m,n)$ be the remainder when $m$ is divided by $n$. Note that
\[
|I(X_1^m;Z(X_1^m)|s_0)-I(X_1^{kn};Z(X_1^{kn})|s_0)|\le r(m,n)(\log|\mathcal{X}|+\log |\mathcal{Z}|)\le n(\log|\mathcal{X}|+\log |\mathcal{Z}|).
\]
Hence we have
$$
\lim_{m\to\infty}\frac{I(X_1^m;Z(X_1^m)|s_0)}{m}=\lim_{k\to\infty}\frac{I(X_1^{kn};Z(X_1^{kn})|s_0)}{kn}\quad\mbox{exists},
$$
as desired.

{\bf Proof of~(iii):} As the proof for different $j$ is similar, we only prove the claim for $j=1$. Let $X^{(1)tn-1}_{1}=(X^{(1)}_{1},\cdots,X^{(1)}_{tn-1})$ and $Z(X^{(1)tn-1}_{1})$ be the output obtained by passing $X^{(1)tn-1}_{1}$ through the channel $W^{tn-1}$.
Using Lemma~\ref{sideinformation}, we have that 
\begin{align}\label{im311}
|I(X^{(1)tn-1}_{1};Z(X^{(1)tn-1}_{1})|s_0)-I(X^{(1)tn-1}_{1};Z(X^{(1)n-1}_{1}),Z(X^{(1)tn-1}_{n})|s_0)|\le \log n.
\end{align}
Note that 
\begin{align}\label{im312}
&I(X^{(1)tn-1}_{1};Z(X^{(1)n-1}_{1}),Z(X^{(1)tn-1}_{n})|s_0)\notag\\
&=I(X^{(1)n-1}_{1};Z(X^{(1)n-1}_{1}),Z(X^{(1)tn-1}_{n}))|s_0)+I(X^{(1)tn-1}_{n};Z(X^{(1)n-1}_{1}),Z(X^{(1)tn-1}_{n})|s_0,X^{(1)n-1}_{1})\notag\\
&\stackrel{(a)}{=}I(X^{(1)n-1}_{1};Z(X^{(1)n-1}_{1}),Z(X^{(1)tn-1}_{n}))|s_0)+I(X^{(1)tn-1}_{n};Z(X^{(1)tn-1}_{n})|s_0,X^{(1)n-1}_{1}),
\end{align}
where~$(a)$ follows from the fact that given $(S_0,X^{(1)n-1}_{1})$, $(X^{(1)tn-1}_{n},Z(X^{(1)tn-1}_{n}))$ is conditionally independent of $Z(X^{(1)n-1}_{1})$. 
As 
{\small$$|I(X^{(1)tn-1}_{n};Z(X^{(1)tn-1}_{n})|s_0,X^{(1)n-1}_{1})-I(X^{(1)tn-1}_{n};Z(X^{(1)tn-1}_{n})|s_0,X^{(1)n-1}_{1},S_{n-1-N(n-1)})|\le \log|\mathcal{S}|$$}
and
$X^{(1)tn-1}_{n}$ is independent of $X^{(1)n-1}_{1}$,
$$|I(X^{(1)tn-1}_{n};Z(X^{(1)tn-1}_{n})|s_0,X^{(1)n-1}_{1})-I(X^{(1)tn-1}_{n};Z(X^{(1)tn-1}_{n})|S_{n-1-N(n-1)})|\le \log|\mathcal{S}|.$$
As $p_{X^{(1)tn-1}_{n}|S_{n-1-N(n-1)}}(\cdot|s)=p_{X_1^{(t-1)n}}(\cdot)$ for any $s$, it follows from Proposition~\ref{blockdifferenceindependent} that for any $\eps>0$ and $ (t-1)n\ge k\ge K(\eps)$,
\begin{align}\label{im313}
&|I(X^{(1)tn-1}_{n};Z(X^{(1)tn-1}_{n})|S_{n-1-N(n-1)})-I(X_1^{(t-1)n};Z(X_1^{(t-1)n})|s_0)|\notag\\
&\le \sum_{s}p_{S_{n-1-N(n-1)}}(s)|I(X^{(1)tn-1}_{n};Z(X^{(1)tn-1}_{n})|S_{n-1-N(n-1)}=s)-I(X_1^{(t-1)n};Z(X_1^{(t-1)n})|s_0)|\notag\\
&\le 2(\log|\mathcal{S}|+\log(k+1)+k\log|\mathcal{X}|)+((t-1)n-k)\eps\log|\mathcal{X}|.
\end{align}
Combining~(\ref{im311}),~(\ref{im312}), and~(\ref{im313}), we obtain that
{\small\begin{align}
&|I(X^{(1)};Z(X^{(1)})|s_0)-I(X;Z(X)|s_0)|\notag\\
&\le\lim_{t\to\infty}\frac{|I(X^{(1)tn-1}_{1};Z(X^{(1)tn-1}_{1})|s_0)-I(X^{(t-1)n}_{1};Z(X^{(t-1)n-1}_{1})|s_0)|}{nt-1}\notag\\
&\le \eps\log|\mathcal{X}|.
\end{align}}
As $\eps$ is arbitrary, we have $I(X^{(1)};Z(X^{(1)})|s_0)=I(X;Z(X)|s_0)$, as desired.

{\bf Proof of~(iv):} Note that 
\begin{align}
&I({X}_1^{kn};Z({X}_1^{n}),\cdots,Z({X}_{(k-1)n+1}^{kn})|s_0)\notag\\
&=\sum_{i=1}^{k}I({X}_{(i-1)n+1}^{in};Z({X}_1^{n}),\cdots,Z({X}_{(k-1)n+1}^{kn})|\hat{X}_{1}^{(i-1)n},s_0)\notag\\
&\stackrel{(b)}{=}\sum_{i=1}^{k}I({X}_{(i-1)n+1}^{in};Z({X}_{(i-1)n+1}^{in}),\cdots,Z({X}_{(k-1)n+1}^{kn})|\hat{X}_{1}^{(i-1)n},s_0)\notag\\
&\stackrel{}{=}\sum_{i=1}^{k}\left\{I({X}_{(i-1)n+1}^{in};Z({X}_{(i-1)n+1}^{in})|\hat{X}_{1}^{(i-1)n},s_0)\right.\notag\\
&\hspace{0.6cm}\left.+I({X}_{(i-1)n+1}^{in};Z({X}_{in+1}^{(i+1)n}),\cdots,Z({X}_{(k-1)n+1}^{kn})|\hat{X}_{1}^{(i-1)n},Z({X}_{(i-1)n+1}^{in}),s_0)\right\}\notag,
\end{align}
where~$(b)$ follows from the fact that given $(s_0,\hat{X}_{1}^{(i-1)n})$, $(Z({X}_1^{n}),\cdots,Z({X}_{(i-2)n+1}^{(i-1)n}))$ is conditionally independent of $({X}_{(i-1)n+1}^{in},Z({X}_{(i-1)n+1}^{in}),\cdots,Z({X}_{(k-1)n+1}^{kn}))$ .

As given $(S_{in-N(in)},\hat{X}_{1}^{(i-1)n},Z({X}_{(i-1)n+1}^{in}),s_0)$, $(Z({X}_{in+1}^{(i+1)n}),\cdots,Z({X}_{(k-1)n+1}^{kn}))$ is conditionally independent of ${X}_{(i-1)n+1}^{in}$, we have that
$$I({X}_{(i-1)n+1}^{in};Z({X}_{in+1}^{(i+1)n}),\cdots,Z({X}_{(k-1)n+1}^{kn})|\hat{X}_{1}^{(i-1)n},Z({X}_{(i-1)n+1}^{in}),s_0,S_{in-N(in)})=0,$$
which together with the fact that 
\begin{align}
&|I({X}_{(i-1)n+1}^{in};Z({X}_{in+1}^{(i+1)n}),\cdots,Z({X}_{(k-1)n+1}^{kn})|\hat{X}_{1}^{(i-1)n},Z({X}_{(i-1)n+1}^{in}),s_0,S_{in-N(in)}))\notag\\
&\hspace{.2cm}-I({X}_{(i-1)n+1}^{in};Z({X}_{in+1}^{(i+1)n}),\cdots,Z({X}_{(k-1)n+1}^{kn})|\hat{X}_{1}^{(i-1)n},Z({X}_{(i-1)n+1}^{in}),s_0|\le \log|\mathcal{S}|,
\end{align}
implies that
\begin{align}\label{im317}
&\left |I({X}_1^{kn};Z({X}_1^{n}),\cdots,Z({X}_{(k-1)n+1}^{kn})|s_0)-\sum_{i=1}^{k}I({X}_{(i-1)n+1}^{in};Z({X}_{(i-1)n+1}^{in})|\hat{X}_{1}^{(i-1)n},s_0)\right|\notag\\
&\le \sum_{i=1}^{k}I({X}_{(i-1)n+1}^{in};Z({X}_{in+1}^{(i+1)n}),\cdots,Z({X}_{(k-1)n+1}^{kn})|\hat{X}_{1}^{(i-1)n},Z({X}_{(i-1)n+1}^{in}),s_0)\notag\\
&\le k\log|\mathcal{S}|.
\end{align}
It then follows from (ii) in Corollary~\ref{existenceindependent} that for any $\delta>0$ and $n\ge t\ge K(\delta)$
\begin{align}\label{im316}
&|I({X}_{(i-1)n+1}^{in};Z({X}_{(i-1)n+1}^{in})|S_{(i-1)n-N((i-1)n)})-I({X}_{1}^{n};Z({X}_{1}^{n})|s_0)|\notag\\
&\le 2(\log|\mathcal{S}|+\log(t+1)+t\log|\mathcal{X}|)+(n-t)\delta\log|\mathcal{X}|.
\end{align}
Note that
\begin{align}\label{im315}
&|I({X}_{(i-1)n+1}^{in};Z({X}_{(i-1)n+1}^{in})|{X}_{1}^{(i-1)n},s_0)-I({X}_{(i-1)n+1}^{in};Z({X}_{(i-1)n+1}^{in})|S_{(i-1)n-N((i-1)n)})|\notag\\
&\le |I({X}_{(i-1)n+1}^{in};Z({X}_{(i-1)n+1}^{in})|{X}_{1}^{(i-1)n},S_{(i-1)n-N((i-1)n)},s_0)\notag\\
&\hspace{1cm}-I({X}_{(i-1)n+1}^{in};Z({X}_{(i-1)n+1}^{in})|S_{(i-1)n-N((i-1)n)})|+\log|\mathcal{S}|\notag\\
&\stackrel{(a)}{=}\log|\mathcal{S}|,
\end{align}
where~$(a)$ follows from the fact that given $S_{(i-1)n-N((i-1)n)}$, $({X}_{(i-1)n+1}^{in},Z({X}_{(i-1)n+1}^{in}))$ is conditionally independent of $(\hat{X}_{1}^{(i-1)n},S_0)$. Combining~(\ref{im317}),~(\ref{im316}), and~(\ref{im315}), we obtain that
\begin{align}
&\left|\lim_{k\to\infty}\frac{I({X}_1^{kn};Z({X}_1^{n}),\cdots,Z({X}_{(k-1)n+1}^{kn})|s_0)}{kn}-\frac{I({X}_{1}^{n};Z({X}_{1}^{n})|s_0)}{n}\right|\notag\\
&\le \frac{2(2\log|\mathcal{S}|+\log(t+1)+t\log|\mathcal{X}|)+(n-t)\delta\log|\mathcal{X}|}{n}.
\end{align}
Choosing $N$ such that $\frac{2(2\log|\mathcal{S}|+\log(t+1)+t\log|\mathcal{X}|)}{N}\le\frac{\eps}{2}$ and $\delta$ such that $\delta\log|\mathcal{X}|\le\frac{\eps}{2}$, we then have that for $n\ge N$,
\[
\left|\lim_{k\to\infty}\frac{I({X}_1^{kn};Z({X}_1^{n}),\cdots,Z({X}_{(k-1)n+1}^{kn})|s_0)}{kn}-\frac{I({X}_{1}^{n};Z({X}_{1}^{n})|s_0)}{n}\right|\le\eps.
\]
\end{proof}

\section{Proof of Proposition~\ref{blockdifferencestationary}}\label{blockdifferencestationaryproof}
\begin{proof}
From Lemma~\ref{sideinformation}, it follows that for any $1\le t\le n$
{\small\begin{align}\label{im01}
&\frac{|I(X_1^n;Z(X_1^n)|s_0)-I(X_{k+1}^{k+n};Z(X_{k+1}^{k+n}))|}{n}\notag\\
&\hspace{1cm}\le \frac{|I(X_1^n;Z(X_1^t),Z(X_{t+1}^n)|s_0)-I(X_{k+1}^{k+n};Z(X_{k+1}^{k+t}),Z(X_{k+t+1}^{k+n}))|}{n}+\frac{2\log(t+1)}{n}\notag\\
&\hspace{1cm}\le \frac{|I(X_{t+1}^n;Z(X_{t+1}^n)|s_0)-I(X_{k+t+1}^{k+n};Z(X_{k+t+1}^{k+n}))|}{n}+\frac{2(t(\log|\mathcal{X}|+\log|\mathcal{Z}|)+\log(t+1))}{n}
\end{align}}

%Let $N_1$ be the number of deletions occured during the transmission of $X_1^t$ and $N_2$ be the number of deletions occurred during the transmission of $X_{k+1}^{k+t}$. Let $N$ be the number of deletions during the transmission of $X_1^k$. Then $N_1$ and $N_2$ are independent and have the same distribution. 

Note that for any $x_{t+1}^n$ and $z_2^*\in \mathcal{Z}^*_{n-t}$
{\small\begin{align}
&\hspace{-1cm}\Pr(X_{t+1}^n=x_{t+1}^n,Z(X_{t+1}^n)=z_2^*|s_0)\notag\\
&\hspace{-1cm}\stackrel{}{=}\sum_{s}\Pr(X_{t+1}^{n}=x_{t+1}^n, S_{t-N(t)}=s|s_0)\Pr(Z(X_{t+1}^n)=z_2^*|X_{t+1}^n=x_{t+1}^n, S_{t-N(t)}=s)\notag\\
&\hspace{-1cm}\stackrel{(a)}{=}\sum_{s,\hat{x}_1^t}P_{X_1^n}(\hat{x}_1^tx_{t+1}^n) P_{S_{t-N(t)}|X_1^t,S_0}(s|\hat{x}_1^t,s_0)\Pr(Z(X_{t+1}^n)=z_2^*|X_{t+1}^{n}=x_{t+1}^n, S_{t-N(t)}=s)\notag
\end{align}}
where~$(a)$ follows from the conditional independence of $S_{t-N_1}$ and $X_{t+1}^n$ given $(S_0,X_{1}^{t})$. Similarly, we have that
\begin{align}
&\Pr(X_{k+t+1}^{k+n}=x_{t+1}^n,Z(X_{k+t+1}^{k+n})=z_2^*)\notag\\
&\stackrel{}{=}\sum_{s}p_{X_{k+t+1}^{k+n},S_{t+k-N(k)-N_{k+1}(t)}}(x_{t+1}^n,s)  \Pr(Z(X_{k+t+1}^{k+n}=z_2^*|X_{k+t+1}^{k+n}=x_{t+1}^{n}, S_{t+k-N-N_{k+1}(t)}=s)\notag\\
&\stackrel{(b)}{=}\sum_{s,s^\prime,\hat{x}_1^t}\left\{p_{X_{k+1}^{k+n}}(\hat{x}_1^t x_{t+1}^n)p_{S_{k-N(k)}|X_{k+1}^{k+n}}(s^\prime|\hat{x}_1^t x_{k+1}^n) p_{S_{t+k-N(k)-N_{k+1}(t)}|X_1^t,S_{k-N(k)}}(s|\hat{x}_1^t,s^\prime)\right.\notag\\
&\hspace{5cm}\times\left. \Pr(Z(X_{t+1}^{n})=z_2^*|X_{t+1}^{n}=x_{t+1}^{n}, S_{t-N(t)}=s)\right\},
\end{align}
where in~$(b)$ we have used the fact that
\begin{align}
&\Pr(Z(X_{t+1}^n)=z_2^*|X_{t+1}^{n}=x_{t+1}^n, S_{t-N(t)}=s)\notag\\
&\hspace{2cm}=\Pr(Z(X_{t+k+1}^{k+n})=z_2^*|X_{t+k+1}^{k+n}=x_{t+1}^{n}, S_{t+k-N(k)-N(t)}=s).
\end{align}
From Lemma~\ref{indecomposable} it then follows that for $t\ge K(\eps)$
\begin{align}
|p_{S_{t-N(t)}|X_1^t,S_0}(s|x_1^t,s_0)-p_{S_{t+k-N(k)-N_{k+1}(t)}|X_1^t,S_{k-N(k)}}(s|x_1^t,s^\prime)|\le \eps,
\end{align}
which, together with the linearity of expectations, further implies that
\begin{align}
&|\Pr(X_{k+t+1}^{k+n}=x_{t+1}^n,Z(X_{k+t+1}^{k+n})=z_2^*)-\Pr(X_{t+1}^n=x_{t+1}^n, Z(X_{t+1}^n)=z_2^*|s_0)|\notag\\
&\le\sum_{s,s^\prime,\hat{x}_1^t}\left\{p_{X_{k+1}^{k+n}}(\hat{x}_1^t x_{t+1}^n)p_{S_{k-N(k)}|X_{k+1}^{k+n}}(s^\prime|\hat{x}_1^t x_{k+1}^n) \Pr(Z(X_{t+1}^{n})=z_2^*|X_{t+1}^{n}=x_{t+1}^{n}, S_{t-N(t)}=s)\right.\notag\\
&\hspace{5cm}\times\left.|p_{S_{t-N(t)}|X_1^t,S_0}(s|x_1^t,s_0)-p_{S_{t+k-N(k)-N_2}|X_1^t,S_{k-N(k)}}(s|x_1^t,s^\prime)| \right\}\notag\\
&\le\sum_{s,s^\prime,\hat{x}_1^t}p_{X_{k+1}^{k+n}}(\hat{x}_1^t x_{t+1}^n)p_{S_{k-N(k)}|X_{k+1}^{k+n}}(s^\prime|\hat{x}_1^t x_{k+1}^n) \Pr(Z(X_{t+1}^{n})=z_2^*|X_{t+1}^{n}=x_{t+1}^{n}, S_{t-N(t)}=s)\eps\notag\\
&=\sum_{s}p_{X_{k+t+1}^{k+n}}(x_{t+1}^n)\Pr(Z(X_{t+1}^{n})=z_2^*|X_{t+1}^{n}=x_{t+1}^{n}, S_{t-N(t)}=s)\eps.
\end{align}
For any two probability mass functions $p(\cdot),q(\cdot)$ over some finite alphabet $\mathcal{U}$, we use $\| p-q \|$ to denote $\sum_{u\in\mathcal{U}}|p(u)-q(u)|$. Then
\begin{align}
 \| p_{X_{k+t+1}^{k+n},Z(X_{k+t+1}^{k+n})}(\cdot,\cdot)-p_{X_{t+1}^n, Z(X_{t+1}^n)|s_0}(\cdot,\cdot|s_0) \|\le |\mathcal{S}|\eps.
\end{align}
Similarly,
$$
\| p_{Z(X_{k+t+1}^{k+n})}(\cdot)-p_{ Z(X_{t+1}^n)|s_0}(\cdot|s_0)\|\le |\mathcal{S}|\eps.
$$
Using~{\cite[Theorem 17.3.3, pp.~664]{Cover}}, we have that
{\small\begin{align}\label{im2}
&\frac{|H(X_{t+1}^n,Z(X_{t+1}^n)|s_0)-H(X_{k+t+1}^{k+n},Z(X_{k+t+1}^{k+n}))|}{n}\stackrel{(a)}{\le} \eps|\mathcal{S}|(\log|\mathcal{X}|+\log|\mathcal{Z}|)-\frac{\eps|\mathcal{S}|\log(|\mathcal{S}|\eps)}{n},
\end{align}}
where in~$(a)$ we have used the fact that the total number of possible values of $(X_{t+1}^n,Z(X_{t+1}^n))$ is less than $|\mathcal{Z}|^{n-t+1}|\mathcal{X}|^{n-t}$.
Similarly, we have that
\begin{align}\label{im3}
\frac{|H(Z(X_{t+1}^n)|s_0)-H(Z(X_{k+t+1}^{k+n}))|}{n}\le \eps|\mathcal{S}|\log|\mathcal{Z}|-\frac{\eps|\mathcal{S}|\log(|\mathcal{S}|\eps)}{n}.
\end{align}
Combining~(\ref{im01}),(\ref{im2}) and~(\ref{im3}) together, we obtain that
\begin{align}
&\frac{|I(X_1^n;Z(X_1^n)|s_0)-I(X_{k+1}^{k+n};Z(X_{k+1}^{k+n}))|}{n}\notag\\
&\le \eps|\mathcal{S}|(\log|\mathcal{X}|+2\log|\mathcal{Z}|)+\frac{2(t(\log|\mathcal{X}|+\log|\mathcal{Z}|)+\log(t+1)-\eps|\mathcal{S}|\log(|\mathcal{S}|\eps))}{n},\notag
\end{align}
as desired.
\end{proof}
\section{Proof of Corollary~\ref{initialindependent}}\label{initialindependentproof}
\begin{proof}
The proof of the existence of the limit in $I(X;Y|s_0)$ is similar to that of (ii), so we omit the proof. Let $\delta>0$ and $k\ge K(\delta)$, where $K(\delta)$ is given as in Lemma~\ref{indecomposable}. Let $s_0$ and $s_0^\prime$ be two different initial states and $n\ge k$. Using similar arguments as in the derivation of~(\ref{im1}), we obtain that
\begin{align}\label{im41}
&|I(X_1^n;Z(X_1^n)|s_0)-I(X_1^n;Z(X_1^n)|s_0^\prime)|\le 2[\log(k+1) +k\log |\mathcal{X}|]\notag\\
&\hspace{0.5cm}+|I(X_{k+1}^n;Z(X_{k+1}^n)|Z(X_1^k),X_1^k,s_0)-I(X_{k+1}^n;Z(X_{k+1}^n)|Z(X_1^k),X_1^k,s_0^\prime)|.
\end{align}
Then
{\small\begin{align}\label{im42}
&|I(X_{k+1}^n;Z(X_{k+1}^n)|Z(X_1^k),X_1^k,s_0)-I(X_{k+1}^n;Z(X_{k+1}^n)|Z(X_1^k),X_1^k,s_0^\prime)|\notag\\
&\hspace{1cm}\le |I(X_{k+1}^n;Z(X_{k+1}^n)|S_{k-N(k)},X_1^k,s_0)-I(X_{k+1}^n;Z(X_{k+1}^n)|S_{k-N(k)},X_1^k,s_0^\prime)|+2\log |\mathcal{S}|,
\end{align} }
where~(\ref{im42}) follows from the conditional independence of $Z(X_1^k)$ and $(X_{k+1}^n,Z(X_{k+1}^n)$ given $(S_{k-N},X_1^k,s_0)$.

From the definition of conditional mutual inofrmation we have that
\begin{align}\label{im43}
&I(X_{k+1}^n;Z(X_{k+1}^n)|S_{k-N(k)},X_1^k,s_0)\notag\\
&=\sum_{x_1^k,s}p_{X_1^k}(x_1^k)p_{S_{k-N(k)}|X_1^k,S_0}(s|x_1^k,s_0)I(X_{k+1}^n;Z(X_{k+1}^n)|S_{k-N(k)}=s,X_1^k=x_1^k)
\end{align}
and
\begin{align}\label{im44}
&I(X_{k+1}^n;Z(X_{k+1}^n)|S_{k-N(k)},X_1^k,s_0^\prime)\notag\\
&=\sum_{x_1^k,s}p_{X_1^k}(x_1^k)p_{S_{k-N(k)}|X_1^k,S_0}(s|x_1^k,s_0^\prime)I(X_{k+1}^n;Z(X_{k+1}^n)|S_{k-N(k)}=s,X_1^k=x_1^k).
\end{align}
Then it follows that
{\small\begin{align}\label{im46}
&|I(X_1^n;Z(X_1^n)|s_0)-I(X_1^n;Z(X_1^n)|s_0^\prime)|\notag\\
&\le |I(X_{k+1}^n;Z(X_{k+1}^n)|S_{k-N(k)},X_1^k,s_0)-I(X_{k+1}^n;Z(X_{k+1}^n)|S_{k-N(k)},X_1^k,s_0^\prime)|+2\log |\mathcal{S}|\notag\\
&\le\sum_{x_1^k,s}p_{X_1^k}(x_1^k)|p_{S_{k-N(k)}|X_1^k,S_0}(s|x_1^k,s_0)-p_{S_{k-N(k)}|X_1^k,S_0}(s|x_1^k,s_0^\prime )|I(X_{k+1}^n;Z(X_{k+1}^n)|S_{k-N(k)}=s,X_1^k=x_1^k)\notag\\
&\hspace{1cm}+2[\log(k+1) +k\log |\mathcal{X}|]+2\log |\mathcal{S}|\notag\\
&\le (n-k)\log |\mathcal{X}| \delta+2[\log(k+1) +k\log |\mathcal{X}|]+2\log |\mathcal{S}|,
\end{align}}
which further implies that
\begin{align}
|I(X;Z(X)|s_0)-I(X;Z(X)|s_0^\prime)|&=\lim_{n\to\infty}\frac{|I(X_1^n;Z(X_1^n)|s_0)-I(X_1^n;Z(X_1^n)|s_0^\prime)|}{n}\notag\\
&\le \lim_{n\to\infty}\frac{(n-k)\log |\mathcal{X}| \delta+2[\log(k+1) +k\log |\mathcal{X}|]+2\log |\mathcal{S}|}{n}\notag\\
&=\delta\log|\mathcal{X}|.
\end{align}
As $\delta$ is arbitrary, $I(X;Z(X)|s_0)=I(X;Z(X)|s_0^\prime)$, as desired.
\end{proof}

\section{Proof of Theorem~\ref{shannonindependent}}\label{shannonindependentproof}
\begin{proof}
Let $C_{n}(s_0)=\frac{1}{n}\sup_{p_{X_1^n}(\cdot)}I(X_1^n;Z(X_1^n)|s_0)$. Let $\delta>0$  and $k\ge K(\delta)$, where $K(\delta)$ is given as in Lemma~\ref{indecomposable}. First we show that $C_{\mathrm{Shannon}}$ does not depend the choice of the initial state.  As $I(X_1^n;Z(X_1^n)|s_0)$ is continuous with respect to $p_{X_1^n}(\cdot)$ and the set of probability mass functions is a compact subset of $R^{|\mathcal{X}|^n}$, there exists $p^*_{X_1^n}(\cdot)$ (that depends on $s_0$) that achieves the supremum in $C_{n}(s_0)$. We will use $I_{p_{X}}(X;Y)$ to denote the mutual information between $X$ and $Y$ when the distribution of $X$ is $p_X$. Then from~(\ref{im46}) we obtain that
\begin{align}
C_{n}(s_0)-C_n(s_0^\prime)&\le \frac{I_{p^*_{X_1^n}}(X_1^n;Z(X_1^n)|s_0)-I_{p^*_{X_1^n}}(X_1^n;Z(X_1^n)|s_0^{\prime})}{n}\notag\\
&\le \frac{(n-k)\delta\log |\mathcal{X}| +2[\log(k+1) +k\log |\mathcal{X}|]+2\log |\mathcal{S}|}{n}
\end{align}
By swapping the role of $s_0$ and $s_0^\prime$, we have that
\begin{align}
C_{n}(s_0^\prime)-C_n(s_0)
&\le \frac{(n-k)\delta\log |\mathcal{X}| +2[\log(k+1) +k\log |\mathcal{X}|]+2\log |\mathcal{S}|}{n}.
\end{align}
As $C_{\mathrm{Shannon}}(s)=\lim_{n\to\infty}C_{n}(s)$, 
\begin{align}
|C_{\mathrm{Shannon}}(s_0)-C_{\mathrm{Shannon}}(s_0^\prime)|&=\lim_{n\to\infty}|C_{n}(s_0)-C_n(s_0^\prime)|\notag\\
&\le \lim_{n\to\infty}\frac{(n-k)\delta\log |\mathcal{X}| +2[\log(k+1) +k\log |\mathcal{X}|]+2\log |\mathcal{S}|}{n}\notag\\
&=\delta\log |\mathcal{X}|.
\end{align}
Due to the arbitrariness of $\delta$, $C_{\mathrm{Shannon}}(s_0)=C_{\mathrm{Shannon}}(s_0^\prime)$, that is, $C_{\mathrm{Shannon}}$ does not depend on the choice of the initial state.

Now we prove that the limit in the definition of $C_{\mathrm{Shannon}}$ exists. 
Fix $r, t \geq 0$, and let $p^*$ and $q^*$ be input distributions that achieve $C_{r}(s_0)$ and $C_{t}(s_0)$, respectively. From now on, we assume that
\begin{equation} \label{p-q}
(X_1^{r},X_{r+1}^{r+t}) \sim p^*(x_1^r) \times q^*(x_{r+1}^{r+t});
\end{equation}
in other words, $X_1^r$ and $X_{r+1}^{r+t}$ are independent and distributed according to $p^*$ and $q^*$, respectively. Note that
$$
|I(X_{1}^{r+t};Z(X_{1}^{r+t})|s_0)-I(X_{1}^{r+t};Z(X_{1}^{r}),Z(X_{r+1}^{r+t})|s_0)|\le \log(r+1)
$$
and
\begin{align}
I(X_{1}^{r+t};Z(X_{1}^{r}),Z(X_{r+1}^{r+t})|s_0)&\ge I(X_{1}^{r};Z(X_{1}^{r})|s_0)+I(X_{r+1}^{r+t};Z(X_{r+1}^{r+t})|s_0,X_{1}^{r})\notag\\
&=rC_r(s_0)+I(X_{r+1}^{r+t};Z(X_{r+1}^{r+t})|s_0,X_{1}^{r}).
\end{align}
 Then using~(\ref{p-q}), we have that  
\begin{align}
I(X_{r+1}^{r+t};Z(X_{r+1}^{r+t})|s_0,X_{1}^{r})&=\sum_{x_1^r,s}p_{X_1^r}(x_1^r)p_{S_{r-N(r)}|X_1^r,S_0}(s|x_1^r,s_0)I(X_{r+1}^n;Z(X_{r+1}^{r+t})|S_{r-N(r)}=s,X_1^r=x_1^r)\notag\\
&=\sum_{x_1^r,s}p_{X_1^r}(x_1^r)p_{S_{r-N(r)}|X_1^r,S_0}(s|x_1^r,s_0)I(X_{r+1}^n;Z(X_{r+1}^{r+t})|S_{r-N(r)}=s)\notag\\
&=\sum_{x_1^r,s}p_{X_1^r}(x_1^r)p_{S_{r-N(r)}|X_1^r,S_0}(s|x_1^r,s_0)I(X_{1}^t;Z(X_{1}^t)|S_{0}=s),
\end{align}
which, together with~(\ref{im46}), implies that
\begin{align}
&|I(X_{r+1}^{r+t};Z(X_{r+1}^{r+t})|s_0,X_{1}^{r})-I(X_{1}^{t};Z(X_{1}^{t})|s_0)|\notag\\
&\le \sum_{x_1^r,s}P_{X_1^r}(x_1^r)p_{S_{r-N(r)}|X_1^r,S_0}(s|x_1^r,s_0)(t-k)\log |\mathcal{X}| \delta+2[\log(k+1) +k\log |\mathcal{X}|]+2\log |\mathcal{S}|\notag\\
&=(t-k)\delta\log |\mathcal{X}| +2[\log(k+1) +k\log |\mathcal{X}|]+2\log |\mathcal{S}|.
\end{align}

Therefore,
\begin{align}
(r+t)C_{r+t}(s_0)&\ge I(X_{1}^{r+t};Z(X_1^{r+t})|s_0)\nonumber\\
&\ge I(X_{1}^{r+t};Z(X_1^{r}),Z(X_{r+1}^{r+t}))-\log(r+1)\nonumber\\
&\ge rC_r(s_0)+I(X_{r+1}^{r+t};Z(X_{r+1}^{r+t})|s_0,X_{1}^{r})-\log(r+1)\notag\\
&\ge rC_r(s_0)+I(X_{1}^{t};Z(X_{1}^{t})|s_0)\nonumber\\
&{}\hspace{4.5mm}-\log(r+1)-(t-k)\log |\mathcal{X}| \delta+2[\log(k+1) +k\log |\mathcal{X}|]+2\log |\mathcal{S}|\notag\\
&\ge rC_r(s_0)+tC_t(s_0)\nonumber\\
&{}\hspace{4.5mm}-\log(r+1)-\delta(t-k)\log |\mathcal{X}| +2[\log(k+1) +k\log |\mathcal{X}|]+2\log |\mathcal{S}|\notag\\
\end{align}
Therefore,
\begin{align}\label{super additive}
C_{r+t}(s_0)&\ge \frac{r}{r+t}C_{\mathrm{S}}(s_0)+\frac{t}{r+t}C_{t}(s_0)\nonumber\\
&\hspace{4.5mm}-\frac{\log(r+1)+(t-k)\delta\log |\mathcal{X}| +2[\log(k+1) +k\log |\mathcal{X}|]+2\log |\mathcal{S}|}{r+t}.
%&\ge \frac{s+1}{s+1+t+1}C_{s+1}+\frac{t+1}{s+1+t+1}C_{t+1}-\eps
\end{align}
Let $t>k$ be such that
$$
\frac{\log(r+1)+(t-k)\delta\log |\mathcal{X}| +2[\log(k+1) +k\log |\mathcal{X}|]+2\log |\mathcal{S}|}{r+t}\le \delta.
$$
Then for $k$ and $t$ chosen above, we obtain that
$$
(r+t)\left\{C_{r+t}(s_0)-\frac{\delta}{r+t}\right\} \ge r \left\{C_{r}(s_0)-\frac{\delta}{s}\right\}+t\left\{C_{t}(s_0)-\frac{\delta}{t}\right\}.
$$
By~\cite[Lemma 2, pp.~112]{gallagerbook}, $\lim\limits_{n\to \infty}\left\{C_{n}-\delta/n\right\}$ exists and furthermore
$$
C_{\mathrm{Shannon}}=\lim_{n\to \infty}C_{n}(s_0)=\lim_{n\to \infty}\left\{C_{n}(s_0)-\frac{\delta}{n}\right\}=\sup_{n}\left\{C_{n}(s_0)-\frac{\delta}{n}\right\}=\sup_{n} C_n(s_0).
$$
The proof of the theorem is then complete.

\end{proof}
\section{Proof of Theorem~\ref{smbfsc}}\label{smbfscproof}
From~\cite[Theorem 8]{markovams} it follows that both $(X,Y)$ and $Y$ are asymptotically mean stationary. From~\cite[Corollary 3 and Lemma 3]{grayergodic} it follows that both $(X,Y)$ and $Y$ are ergodic. Then it follows from~\cite[Theorem 3]{barron} that
{\small\begin{align}
-\frac{\log p_{Y_1^n}(Y_1^n)}{n}\to H(Y)\quad \mbox{a.s. and in $L^1$}
\end{align}}
and 
{\small\begin{align}
-\frac{\log W_2(Y_1^n|X_1^n)}{n}\to H(Y|X)\quad \mbox{a.s. and in $L^1$},
\end{align}}
as desired.

\section{Proof that $\hat{W}$ is indecomposable}\label{indecomposablenew}
Let $\eps>0$ and let $(s_0^*,s_0^{*\prime})$ be a pair of initial states. For $i\ge 1$, let $\X_i$ and $\Z_i$ be respectively $\mathcal{X}^n$-valued and $\mathcal{Z}^n$-valued random vectors. Let $\X=\{\X_i\}$ be the input of the channel $\hat{W}$ and $\Z=\{\Z_i\}$ be its corresponding output. Let 
{\small$$\Pr(S_k^{*}=s_k^*|\X_i^k=\zeta_1^k,S_0^{*}=s_0^*)=\sum_{s_1^{k-1*},z_1^{*k}}\hat{W}^{k}(z_1^{*k},s_1^{*k}|\zeta_1^{k},s_0^*).$$To prove that $\hat{W}$ is indecomposable, it suffices to show that there is an integer $K_3$ such that if $k\ge K_3$, then
{\small\begin{align}\label{indecom200}
|\Pr(S_k^{*}=s_k^*|\X_i^k=\zeta_1^k,S_0^{*}=s_0^*)-\Pr(S_k^{*}=s_k^*|\X_1^k=\zeta_1^k,S_0^{*}=s_0^*)|\le \eps,
\end{align}}
for any $\zeta_1^k$ and $s_k^*\in\mathcal{S}^{*n}-\{\emptyset\}$. 

Let $N_k=|\{i:Z_i\not=\emptyset\}|$. Then $N_k$ is a binomial random variable with $k$ trials and success probability $1-d^n$. Note that
\begin{align}\label{def1}
&\Pr(S_k^{*}=s_k^*|\X_1^k=\zeta_1^k,S_0^{*}=s_0^*)=\sum_{m=0}^{k}\Pr(S_k^{*}=s_k^*,N_k=m|\X_1^k=\zeta_1^k,S_0^{*}=s_0^*)\notag\\
&\hspace{3.2cm}=\sum_{m=0}^{k}\sum_{i_1^{m}}\sum_{x_{i_1^{m}}^{*}:\ell(x_{i_m}^{*})=\ell(s_k^*)}\left(\prod_{j=1}^{m}W_1^{n}(x_{i_j}^{*}|\zeta_{i_m})\right)W_2^{\ell(x_{i_1^{m}}^{*})}(s_k^*|x_{i_1^{m}}^{*},s_{0,\ell(s_0^*)}),
\end{align}
where $i_1^m$ is an $m$-dimensional integer vector with $1\le i_j\le k$ and $x_{i_1^{m}}^{*}=(x_{i_1}^{*},\cdots,x_{i_{m}}^{*})$ with $x_{i_{j}}^{*}\in\mathcal{X}^{*n}$ for $1\le j\le m$.
As $W_2$ is an indecomposable FSC, there exists an integer $K_4$ such that if $n\ge K_4$,
{\small$$
\sup_{s}|\Pr(S_n=s|X_1^n=x_1^n,S_0=s_0)-\Pr(S_n=s|X_1^n=x_1^n,S_0=s_0^\prime)|\le \eps,
$$}
for any pair of initial states $(s_0,s_0^\prime)$ and $x_1^n$. Thus it follows that if $\ell(x_{i_1^{m-1}}^{*})\ge K_4$, then
{\small\begin{align}\label{indecom100}
&|W_2^{\ell(x_{i_1^{m}}^{*})}(s_k^*|x_{i_1^{m}}^{*},s_{\ell(s_0^*)}^{*})-W_2^{\ell(x_{i_1^{m}}^{*})}(s_k^*|x_{i_1^{m}}^{*},s_{\ell(s_0^{*\prime})}^{*\prime})|\notag\\
&=p(s_{k,2}^{k,\ell(s_k^{*})}|x_{i_{m},2}^{i_{m},\ell(x_{i_m}^{*})},s_{k,1})\notag\\
&\hspace{0.5cm}\times\left|W_2^{\ell\left(x_{i_1^{m-1}}^{*}x_{i_{m},1}\right)}(s_{k,1}|x_{i_1^{m}}^{*},s_{0,\ell(s_0^{*})})-W_2^{\ell\left(x_{i_1^{m-1}}^{*}x_{i_{m},1}\right)}(s_{k,1}|x_{i_1^{m-1}}^{*}x_{i_{m},1},s_{0,\ell(s_0^{*\prime})}^{\prime})\right|\notag\\
&\le \frac{\eps}{2},
\end{align}}
where $x_{i_1^{m-1}}^{*}x_{i_{m}}$ is the concatenation of its symbols,  $x_{i_{m},2}^{k,\ell(x_{i_{m}}^{*})}=(x_{i_{m},2},\cdots,x_{i_{m},\ell(x_{i_{m}}^{*})})$ and $s_{k,2}^{k,\ell(s_k^{*})}=(s_{k,2},\cdots,s_{k,\ell(s_k^{*})})$, $s_{k,j}$ and $x_{i_{m},j}$ are respectively the $j^{\mathrm{th}}$-component of $s_k^*$.
 
Note that 
{\small\begin{align}\label{im700}
&\sum_{m=0}^{K_4}\sum_{i_1^{m}}\sum_{x_{i_1^{m}}^{*}:\ell(x_{i_m}^{*})=\ell(s_k^*)}  \left(\prod_{j=1}^{m}W_1^{n}(x_{i_j}^{*}|\zeta_{i_m})\right)W_2^{\ell(x_{i_1^{m}}^{*})}(s_k^*|x_{i_1^{m}}^{*},s_{\ell(s_0^*)}^{*})\notag\\
&=\Pr(S_k^{*}=s_k^*, N_k\le K_4|\X_1^k=\zeta_1^k,S_0^{*}=s_0^*)\notag\\
&\le \Pr( N_k\le K_4|\X_1^k=\zeta_1^k,S_0^{*}=s_0^*)\stackrel{(a)}{\le} \Pr(N_k\le K_4),
\end{align}}where~$(a)$ follows from the fact that $N_k$ is independent of $(\X_1^k,S_0^{*})$. As $\ell(x_{i_1^{m-1}}^{*})\ge m-1$, then it follows from~(\ref{def1}) and~(\ref{indecom100}) that
{\small\begin{align}\label{indecom100}
&|\Pr(S_k^{*}=s_k^*|\X_1^k=\zeta_1^k,S_0^{*}=s_0^*)-\Pr(S_k^{*}=s_k^*|\X_1^k=\zeta_1^k,S_0^{*}=s_0^{*\prime})|\notag\\
&=\sum_{m=0}^{k}\sum_{i_1^{m}}\sum_{x_{i_1^{m}}^{*}:\ell(x_{i_m}^{*})=\ell(s_k^*)}\left(\prod_{j=1}^{m}W_1^{n}(x_{i_j}^{*}|\zeta_{i_m})\right)|W_2^{\ell(x_{i_1^{m}}^{*})}(s_k^*|x_{i_1^{m}}^{*},s_{\ell(s_0^*)}^{*})-W_2^{\ell(x_{i_1^{m}}^{*})}(s_k^*|x_{i_1^{m}}^{*},s_{\ell(s_0^{*\prime})}^{*\prime})|\notag\\
&=\left\{\sum_{m=0}^{K_4}+\sum_{m=K_4+1}^{k}\right\}\sum_{i_1^{m}}\sum_{x_{i_1^{m}}^{*}:\ell(x_{i_m}^{*})=\ell(s_k^*)} \bigg\{ \left(\prod_{j=1}^{m}W_1^{n}(x_{i_j}^{*}|\zeta_{i_m})\right)\notag\\
&\hspace{7cm}\times|W_2^{\ell(x_{i_1^{m}}^{*})}(s_k^*|x_{i_1^{m}}^{*},s_{\ell(s_0^*)}^{*})-W_2^{\ell(x_{i_1^{m}}^{*})}(s_k^*|x_{i_1^{m}}^{*},s_{\ell(s_0^{*\prime})}^{*\prime})|\bigg\}\notag\\
&\le \sum_{m=0}^{K_4}\sum_{i_1^{m}}\sum_{x_{i_1^{m}}^{*}:\ell(x_{i_m}^{*})=\ell(s_k^*)}  \left(\prod_{j=1}^{m}W_1^{n}(x_{i_j}^{*}|\zeta_{i_m})\right)|W_2^{\ell(x_{i_1^{m}}^{*})}(s_k^*|x_{i_1^{m}}^{*},s_{\ell(s_0^*)}^{*})-W_2^{\ell(x_{i_1^{m}}^{*})}(s_k^*|x_{i_1^{m}}^{*},s_{\ell(s_0^{*\prime})}^{*\prime})|+\frac{\eps}{2}\notag\\
&\stackrel{(a)}{\le} 2\Pr( N_k\le K_4)+\frac{\eps}{2},
\end{align}}
where~$(a)$ follows from~(\ref{im700}).
As $N_k$ is a binomial random variable, there exists an integer $K_3$ such that $2\Pr( N_k\le K_4)\le \eps/2$ for $k\ge K_3$. Thus for $k\ge K_3$, we have~(\ref{indecom200}), as desired.

\end{document}